\documentclass{qtds}
\usepackage{amssymb}
\usepackage{amsmath}
\usepackage{latexsym}
\usepackage{graphicx}
\usepackage{epsfig}

\def\be{\begin{equation}}
\def\ee{\end{equation}}
\def\bq{\begin{eqnarray}}
\def\eq{\end{eqnarray}}
\def\beq{\begin{eqnarray*}}
\def\eeq{\end{eqnarray*}}

\begin{document}

%





\authorrunninghead{J.M. Ginoux and J. Llibre}
\titlerunninghead{Canards Existence in $\mathbb{R}^{2+2}$}

\title{Canards Existence in FitzHugh-Nagumo\\ and \\ Hodgkin-Huxley Neuronal Models}

\author{Jean-Marc Ginoux\thanks{I would like to thank the Universitat
Aut\`onoma de Barcelona for its kind invitation that allowed this
paper to be written.}}
\affil{Laboratoire LSIS, CNRS, UMR 7296, Universit\'{e} de Toulon,\\
 BP 20132, F-83957 La Garde cedex, France.}
\email{ginoux@univ-tln.fr}
\and

\author{Jaume Llibre\thanks{Partially supported by a DGES grant number PB96--1153.}}
\affil{Departament de Matem\`{a}tiques, Universitat Aut\`{o}noma de Barcelona,\\
08193 -- Bellaterra, Barcelona, Spain.}
\email{jllibre@mat.uab.es}

\vspace{24pt}
{\begin{minipage}{24pc}
\footnotesize{\ \ \ \ In a previous paper we have proposed a new method for proving the existence of ``canard solutions'' for three and four-dimensional singularly perturbed systems with only one \textit{fast} variable which improves the methods used until now. The aim of this work is to extend this method to the case of four-dimensional singularly perturbed systems with two \textit{slow} and two \textit{fast} variables. This method enables to state a unique generic condition for the existence of ``canard solutions'' for such four-dimensional singularly perturbed systems which is based on the stability of \textit{folded singularities} (\textit{pseudo singular points} in this case) of the \textit{normalized slow dynamics} deduced from a well-known property of linear algebra. This unique generic condition is identical to that provided in previous works. Applications of this method to the famous coupled FitzHugh-Nagumo equations and to the Hodgkin-Huxley model enables to show the existence of ``canard solutions'' in such systems.}
\end{minipage}}
\vspace{10pt}


\begin{article}

\section{Introduction}

In the beginning of the eighties, Beno\^{i}t and Lobry \cite{BenoitLobry}, Beno\^{i}t \cite{Benoit1983} and then Beno\^{i}t \cite{Benoit1984} in his PhD-thesis studied canard solutions in $\mathbb{R}^3$. In the article entitled ``Syst\`{e}mes lents-rapides dans $\mathbb{R}^3$ et leurs canards,''  Beno\^{i}t \cite[p. 170]{Benoit1983} proved the existence of canards solution for three-dimensional singularly perturbed systems with two \textit{slow} variables and one \textit{fast} variable while using ``Non-Standard Analysis''according to a theorem which stated that canard solutions exist in such systems provided that the \textit{pseudo singular point}\footnote{This concept has been originally introduced by Jos\'{e} Arg\'{e}mi \cite{Argemi}. See Sec. 1.8.} of the \textit{slow dynamics}, \textit{i.e.}, of the \textit{reduced vector field} is of \textit{saddle} type. Nearly twenty years later, Szmolyan and Wechselberger \cite{SzmolyanWechselberger2001} extended ``Geometric Singular Perturbation Theory\footnote{See Fenichel \cite{Fen1971,Fen1979}, O'Malley \cite{OMalley1974}, Jones \cite{Jones1994} and Kaper \cite{Kaper1999}}'' to canards problems in $\mathbb{R}^3$ and provided a ``standard version'' of Beno\^{i}t's theorem \cite{Benoit1983}. Very recently, Wechselberger \cite{Wechselberger2012} generalized this theorem for $n$-dimensional singularly perturbed systems with $k$ \textit{slow} variables and $m$ \textit{fast} (Eq. (\ref{eq1})). The method used by Szmolyan and Wechselberger \cite{SzmolyanWechselberger2001} and Wechselberger \cite{Wechselberger2012} require to implement a ``desingularization procedure'' which can be summarized as follows: first, they compute the \textit{normal form} of such singularly perturbed systems which is expressed according to some coefficients ($a$ and $b$ for dimension three and $\tilde{a}$, $\tilde{b}$ and $\tilde{c}_j$ for dimension four) depending on the functions defining the original vector field and their partial derivatives with respect to the variables. Secondly, they project the ``desingularized vector field'' (originally called ``normalized slow dynamics'' by Eric Beno\^{i}t \cite[p. 166]{Benoit1983}) of such a \textit{normal form} on the tangent bundle of the critical manifold. Finally, they evaluate the Jacobian of the projection of this ``desingularized vector field'' at the \textit{folded singularity} (originally called \textit{pseudo singular points} by Jos\'{e} Arg\'{e}mi \cite[p. 336]{Argemi}). This lead Szmolyan and Wechselberger \cite[p. 427]{SzmolyanWechselberger2001} and Wechselberger \cite[p. 3298]{Wechselberger2012} to a ``classification of \textit{folded singularities} (\textit{pseudo singular points})''. Thus, they showed that for three-dimensional singularly perturbed systems such \textit{folded singularity} is of \textit{saddle type} if the following condition is satisfied: $a<0$ while for four-dimensional singularly perturbed systems such \textit{folded singularity} is of \textit{saddle type} if $\tilde{a}<0$. Then, Szmolyan and Wechselberger \cite[p. 439]{SzmolyanWechselberger2001} and Wechselberger \cite[p. 3304]{Wechselberger2012} established their Theorem 4.1. which state that ``\textit{In the folded saddle and in the folded node case singular canards perturb to maximal canard for sufficiently small $\varepsilon$}''. However, in their works neither Szmolyan and Wechselberger \cite{SzmolyanWechselberger2001} nor Wechselberger \cite{Wechselberger2012} did not provide (to our knowledge) the expression of these constants ($a$ and $\tilde{a}$) which are necessary to state the existence of canard solutions in such systems.

In a previous paper entitled: ``Canards Existence in Memristor's Circuits'' (see Ginoux \& Llibre \cite{GinouxLLibre2015}) we first provided the expression of these constants and then showed that they can be directly determined starting from the \textit{normalized slow dynamics} and not from the projection of the ``desingularized vector field'' of the \textit{normal form}. This method enabled to state a unique ``generic'' condition for the existence of ``canard solutions'' for such three and four-dimensional singularly perturbed systems which is based on the stability of \textit{folded singularities} of the \textit{normalized slow dynamics} deduced from a well-known property of linear algebra. This unique condition which is completely identical to that provided by Beno\^{i}t \cite{Benoit1983} and then by Szmolyan and Wechselberger \cite{SzmolyanWechselberger2001} and finally by Wechselberger \cite{Wechselberger2012} is ``generic'' since it is exactly the same for singularly perturbed systems of dimension three and four with only one \textit{fast} variable.

The aim of this work is to extend this method to the case of four-dimensional singularly perturbed systems with $k=2$ \textit{slow} and $m=2$ \textit{fast} variables. Since the dimension of the system is $m=k+m$, such problem is known as ``canards existence in $\mathbb{R}^{2+2}$''. Moreover, in this particular case where $k=m=2$, the \textit{folded singularities} of Wechselberger \cite[p. 3298]{Wechselberger2012} are nothing else but the \textit{pseudo singular points} of the late Jos\'{e} Arg\'{e}mi \cite{Argemi} as we will see below. Following the previous works, we show that for such four-dimensional singularly perturbed systems \textit{pseudo singular points} are of \textit{saddle type} if $\tilde{a}<0$. Then, according Theorem 4.1. of Wechselberger \cite[p. 3304]{Wechselberger2012} we provide the expression of this constant $\tilde{a}$ which is necessary to establish the existence of canard solutions in such systems. So, we can state that the condition $\tilde{a}<0$ for existence of canards in such $\mathbb{R}^{2+2}$ is ``generic'' since it is exactly the same for singularly perturbed systems of dimension three and four with only one \textit{fast} variable.

The outline of this paper is as follows. In Sec. 1, definitions of singularly perturbed system, critical manifold, reduced system, ``constrained system'', canard cycles, folded singularities and pseudo singular points are recalled. The method proposed in this article is presented in Sec. 2 for the case of four-dimensional singularly perturbed systems with two \textit{fast} variables. In Sec. 3, applications of this method to the famous coupled FitzHugh-Nagumo equations and to the Hodgkin-Huxley model enables to show the existence of ``canard solutions'' in such systems.

\section{Definitions}

\subsection{Singularly perturbed systems}
\label{Sec1}

According to Tikhonov \cite{Tikhonov1948}, Jones \cite{Jones1994} and Kaper \cite{Kaper1999} \textit{singularly perturbed systems} are defined as:

\begin{equation}
\label{eq1}
\begin{aligned}
{\vec {x}}' & = \varepsilon \vec{f} \left( {\vec{x},\vec{y},\varepsilon} \right), \\
{\vec {y}}' & = \vec {g}\left( {\vec{x}, \vec{y},\varepsilon }
\right).
\end{aligned}
\end{equation}

where $\vec {x} \in \mathbb{R}^k$, $\vec {y} \in \mathbb{R}^m$, $\varepsilon \in \mathbb{R}^ + $, and the prime denotes
differentiation with respect to the independent variable $t'$. The functions $\vec {f}$ and $\vec {g}$ are assumed to be $C^\infty$
functions\footnote{In certain applications these functions will be supposed to be $C^r$, $r \geqslant 1$.} of $\vec {x}$, $\vec {y}$
and $\varepsilon$ in $U\times I$, where $U$ is an open subset of $\mathbb{R}^k\times \mathbb{R}^m$ and $I$ is an open interval
containing $\varepsilon = 0$.

\smallskip

In the case when $0 < \varepsilon \ll 1$, \textit{i.e.} $\varepsilon$ is a small positive number, the variable $\vec {x}$ is called \textit{slow} variable, and $\vec {y}$ is called \textit{fast} variable. Using Landau's notation: $O\left( \varepsilon^p \right)$ represents a function $f$ of $u$ and $\varepsilon $ such that $f(u,\varepsilon) / \varepsilon^p$ is bounded for positive $\varepsilon$  going to zero, uniformly for $u$ in the given domain.

\smallskip

In general we consider that $\vec {x}$ evolves at an $O\left( \varepsilon \right)$ rate; while $\vec {y}$ evolves at an $O\left( 1 \right)$ \textit{slow} rate. Reformulating system (\ref{eq1}) in terms of the rescaled variable $t = \varepsilon t'$, we obtain

\begin{equation}
\label{eq2}
\begin{aligned}
\dot {\vec {x}} & = \vec{f} \left( {\vec{x},\vec{y},\varepsilon} \right), \\
\varepsilon \dot {\vec {y}} & = \vec {g}\left( {\vec{x}, \vec{y},\varepsilon }
\right).
\end{aligned}
\end{equation}

The dot represents the derivative with respect to the new independent variable $t$.

\smallskip

The independent variables $t'$ and $t$ are referred to the \textit{fast} and \textit{slow} times, respectively, and (\ref{eq1}) and (\ref{eq2}) are called the \textit{fast} and \textit{slow} systems, respectively. These systems are equivalent whenever $\varepsilon \ne 0$, and they are labeled \textit{singular perturbation problems} when $0 < \varepsilon \ll 1$. The label ``singular'' stems in part from the discontinuous limiting behavior
in system (\ref{eq1}) as $\varepsilon \to 0$.

\smallskip

\subsection{Reduced slow system}

In such case system (\ref{eq2}) leads to a differential-algebraic system (D.A.E.) called \textit{reduced slow system} whose dimension decreases from $k + m = n$ to $m$. Then, the \textit{slow} variable $\vec {x} \in \mathbb{R}^k$ partially evolves in the submanifold $M_0$ called the \textit{critical manifold}\footnote{It represents the approximation of the slow invariant manifold, with an error of $O(\varepsilon)$.}. The \textit{reduced slow system} is

\begin{equation}
\label{eq3}
\begin{aligned}
\dot {\vec {x}} & = \vec{f} \left( {\vec{x},\vec{y},\varepsilon} \right), \\
\vec {0} & = \vec {g}\left( {\vec{x}, \vec{y},\varepsilon }
\right).
\end{aligned}
\end{equation}

\subsection{Slow Invariant Manifold}

The \textit{critical manifold} is defined by

\begin{equation}
\label{eq4} M_0 := \left\{ {\left( {\vec {x},\vec {y}} \right):\vec
{g}\left( {\vec {x},\vec {y},0} \right) = {\vec {0}}} \right\}.
\end{equation}

Such a normally hyperbolic invariant manifold (\ref{eq4}) of the \textit{reduced slow system} (\ref{eq3}) persists as a locally invariant \textit{slow manifold} of the full problem (\ref{eq1}) for $\varepsilon$ sufficiently small. This locally \textit{slow invariant manifold} is $O(\varepsilon)$ close to the \textit{critical manifold}.

When $D_{\vec{x}}\vec{f}$ is invertible, thanks to the Implicit Function Theorem, $M_0 $ is given by the graph of a $C^\infty $ function $\vec {x} = \vec {G}_0 \left( \vec {y} \right)$ for $\vec {y} \in D$, where $D\subseteq \mathbb{R}^k$ is a compact, simply connected domain and the boundary of D is a $(k - 1)$--dimensional $C^\infty$ submanifold\footnote{The set D is overflowing invariant with respect to (\ref{eq2}) when $\varepsilon = 0$. See Kaper \cite{Kaper1999} and Jones \cite{Jones1994}.}.

\smallskip

According to Fenichel \cite{Fen1971,Fen1979} theory if $0 < \varepsilon \ll 1$ is sufficiently small, then there exists a function $\vec {G}\left( {\vec {y},\varepsilon } \right)$ defined on D such that the manifold

\begin{equation}
\label{eq5} M_\varepsilon := \left\{ {\left( {\vec {x},\vec {y}}
\right):\vec {x} = \vec {G}\left( {\vec {y},\varepsilon } \right)} \right\},
\end{equation}

is locally invariant under the flow of system (\ref{eq1}). Moreover, there exist perturbed local stable (or attracting) $M_a$ and unstable (or repelling) $M_r$ branches of the \textit{slow invariant manifold} $M_\varepsilon$. Thus, normal hyperbolicity of $M_\varepsilon$ is lost via a saddle-node bifurcation of the \textit{reduced slow system} (\ref{eq3}). Then, it gives rise to solutions of ``canard'' type.

\subsection{Canards, singular canards and maximal canards}

A \textit{canard} is a solution of a singularly perturbed dynamical system (\ref{eq1}) following the \textit{attracting} branch $M_a$ of the \textit{slow invariant manifold}, passing near a bifurcation point located on the fold of this \textit{slow invariant manifold}, and then following the \textit{repelling} branch $M_r$ of the \textit{slow invariant manifold}.

A \textit{singular canard} is a solution of a \textit{reduced slow system} (\ref{eq3}) following the \textit{attracting} branch $M_{a,0}$ of the \textit{critical manifold}, passing near a bifurcation point located on the fold of this \textit{critical manifold}, and then following the \textit{repelling} branch $M_{r,0}$ of the \textit{critical manifold}.

A \textit{maximal canard} corresponds to the intersection of the attracting and repelling branches $M_{a,\varepsilon} \cap M_{r,\varepsilon}$ of the slow manifold in the vicinity of a non-hyperbolic point.

According to Wechselberger \cite[p. 3302]{Wechselberger2012}:

\begin{quote}
``Such a maximal canard defines a family of canards nearby which are exponentially close to the maximal canard, \textit{i.e.} a family of solutions of (\ref{eq1}) that follow an attracting branch $M_{a,\varepsilon}$ of the slow manifold and then follow, rather surprisingly, a repelling/saddle branch
$M_{r,\varepsilon}$ of the slow manifold for a considerable amount of slow time. The existence of this family of canards is a consequence of the non-uniqueness of $M_{a,\varepsilon}$ and $M_{r,\varepsilon}$. However, in the singular limit $\varepsilon \rightarrow 0$, such a family of canards is represented by a unique singular canard.''
\end{quote}

Canards are a special class of solution of singularly perturbed dynamical systems for which normal hyperbolicity is lost. Canards in singularly perturbed systems with two or more slow variables $(\vec {x} \in \mathbb{R}^k$, $k \geqslant 2)$ and one fast variable $(\vec {y} \in \mathbb{R}^m$, $m = 1)$ are robust, since maximal canards generically persist under small parameter changes\footnote{See Beno\^{i}t \cite{Benoit1983, Benoit2001}, Szmolyan and Wechselberger \cite{SzmolyanWechselberger2001} and Wechselberger \cite{Wechselberger2005,Wechselberger2012}.}.

\subsection{Constrained system}

In order to characterize the ``slow dynamics'',  \textit{i.e.} the slow trajectory of the \textit{reduced slow system} (\ref{eq3})  (obtained by setting $\varepsilon = 0$ in (\ref{eq2})), Floris Takens \cite{Takens1976} introduced the ``constrained system'' defined as follows:

\begin{equation}
\label{eq6}
\begin{aligned}
\dot {\vec {x}} & = \vec{f} \left( {\vec{x},\vec{y},0} \right), \\
D_{\vec{y}} \vec{g} . \dot {\vec {y}} & =  - ( D_{\vec{x}} \vec{g} . \vec {f} ) \left( {\vec{x}, \vec{y}, 0} \right),\\
\vec {0} & = \vec {g}\left( {\vec{x}, \vec{y}, 0 } \right).
\end{aligned}
\end{equation}

\smallskip

Since, according to Fenichel \cite{Fen1971,Fen1979}, the \textit{critical manifold} $\vec{g} \left( {\vec{x},\vec{y},0} \right)$ may be considered as locally invariant under the flow of system (\ref{eq1}), we have:

\[
\hfill \frac{d\vec{g}}{dt} \left( {\vec{x},\vec{y},0} \right) = 0 \quad \Longleftrightarrow \quad D_{\vec{x}} \vec{g} . \dot{\vec{x}} + D_{\vec{y}} \vec{g} . \dot {\vec {y}} = \vec{0}. \hfill
\]

By replacing $\dot{\vec{x}}$ by $\vec{f}\left( {\vec{x},\vec{y},0} \right)$ leads to:

\[
\hfill D_{\vec{x}} \vec{g} . \vec{f}\left( {\vec{x},\vec{y},0} \right) + D_{\vec{y}} \vec{g} . \dot {\vec {y}} = \vec{0}. \hfill
\]

This justifies the introduction of the \textit{constrained system}.

Now, let $adj(D_{\vec{y}} \vec{g})$ denote the adjoint of the matrix $D_{\vec{y}} \vec{g}$ which is the transpose of the co-factor matrix $D_{\vec{y}} \vec{g}$, then while multiplying the left hand side of (\ref{eq6}) by the inverse matrix $(D_{\vec{y}} \vec{g})^{-1}$ obtained by the adjoint method we have:

\begin{equation}
\label{eq7}
\begin{aligned}
\dot {\vec {x}} & = \vec{f} \left( {\vec{x},\vec{y},0} \right), \\
det(D_{\vec{y}} \vec{g}) \dot {\vec {y}} & =  - (adj(D_{\vec{y}} \vec{g}) . D_{\vec{x}} \vec{g} . \vec {f} ) \left( {\vec{x}, \vec{y}, 0} \right),\\
\vec {0} & = \vec {g}\left( {\vec{x}, \vec{y}, 0 } \right).
\end{aligned}
\end{equation}

\subsection{Normalized slow dynamics}

Then, by rescaling the time by setting $t = - det(D_{\vec{y}} \vec{g}) \tau$ we obtain the following system which has been called by Eric Beno\^{i}t \cite[p. 166]{Benoit1983} ``normalized slow dynamics'':

\begin{equation}
\label{eq8}
\begin{aligned}
\dot {\vec {x}} & = - det(D_{\vec{y}} \vec{g}) \vec{f} \left( {\vec{x},\vec{y},0} \right), \\
\dot {\vec {y}} & =  (adj(D_{\vec{y}} \vec{g}) . D_{\vec{x}} \vec{g} . \vec {f} ) \left( {\vec{x}, \vec{y}, 0} \right),\\
\vec {0} & = \vec {g}\left( {\vec{x}, \vec{y}, 0 } \right).
\end{aligned}
\end{equation}

where the overdot now denotes the time derivation with respect to $\tau$.

Let's notice that Jos\'{e} Arg\'{e}mi \cite{Argemi} proposed to rescale time by setting $t = - det(D_{\vec{y}} \vec{g}) sgn(det(D_{\vec{y}} \vec{g})) \tau$ in order to keep the same flow direction in (\ref{eq8}) as in (\ref{eq7}).

\subsection{Desingularized vector field}

By application of the Implicit Function Theorem, let suppose that we can explicitly express from Eq. (\ref{eq4}), say without loss of generality, $x_1$ as a function $\phi_1$ of the other variables. This implies that $M_0$ is locally the graph of a function $\phi_1 \mbox{ : } \mathbb{R}^k \to \mathbb{R}^m$ over the base $U = ( \vec{ \chi }, \vec{y} ) $ where $\vec{ \chi } = (x_2, x_3, ..., x_k)$. Thus, we can span the ``normalized slow dynamics'' on the tangent bundle at the \textit{critical manifold} $M_0$ at the \textit{pseudo singular point}. This leads to the so-called \textit{desingularized vector field}:

\begin{equation}
\label{eq9}
\begin{aligned}
\dot {\vec {\chi}} & = - det(D_{\vec{y}} \vec{g}) \vec{f} \left( {\vec{\chi},\vec{y},0} \right), \\
\dot {\vec {y}} & =  (adj(D_{\vec{y}} \vec{g}) . D_{\vec{x}} \vec{g} . \vec {f} ) \left( {\vec{\chi}, \vec{y}, 0} \right).
\end{aligned}
\end{equation}

\subsection{Pseudo singular points and folded singularities}

As recalled by Guckenheimer and Haiduc \cite[p. 91]{GuckenHaiduc2005}, \textit{pseudo-singular points} have been introduced by the late Jos\'{e} Arg\'{e}mi \cite{Argemi} for low-dimensional singularly perturbed systems and are defined as singular points of the ``normalized slow dynamics'' (\ref{eq8}). Twenty-three years later, Szmolyan and Wechselberger \cite[p. 428]{SzmolyanWechselberger2001} called such \textit{pseudo singular points}, \textit{folded singularities}. In a recent publication entitled ``A propos de canards'' Wechselberger \cite[p. 3295]{Wechselberger2012} proposed to define such singularities for $n$-dimensional singularly perturbed systems with $k$ \textit{slow} variables and $m$ \textit{fast} as the solutions of the following system:

\begin{equation}
\label{eq10}
\begin{aligned}
& det(D_{\vec{y}} \vec{g}) = 0, \\
& (adj(D_{\vec{y}} \vec{g}) . D_{\vec{x}} \vec{g} . \vec {f} ) \left( {\vec{x}, \vec{y}, 0} \right) = \vec{0},\\
& \vec {g}\left( {\vec{x}, \vec{y}, 0 } \right) = \vec {0}.
\end{aligned}
\end{equation}

Thus, for dimensions higher than three, his concept encompasses that of Arg\'{e}mi. Moreover, Wechselberger \cite[p. 3296]{Wechselberger2012} proved that \textit{folded singularities} form a $(k-2)$-dimensional manifold. Thus, for $k=2$ the \textit{folded singularities} are nothing else than the \textit{pseudo singular points} defined by Arg\'{e}mi \cite{Argemi}. While for $k \geqslant 3$ the \textit{folded singularities} are no more points but a $(k-2)$-dimensional manifold. Moreover, let's notice on the one hand that the original system (\ref{eq1}) includes $n=k+m$ variables and on the other hand, that the system (\ref{eq10}) comprises $p=2m +1$ equations. However, in the particular case $k=m=2$, two equations of the system (\ref{eq10}) are linearly dependent. So, such system only comprises$p=2m=2k$ equations. So, all the variables (unknowns) of system (\ref{eq10}) can be determined. The solutions of this system are called \textit{pseudo singular points}. We will see in the next Sec. 2 that the stability analysis of these \textit{pseudo singular points} will give rise to a condition for the existence of canard solutions in the original system (\ref{eq1}).

\section{Four-dimensional singularly perturbed systems with two fast variables}
\label{Sec2}

A four-dimensional \textit{singularly perturbed dynamical system} (\ref{eq2}) with $k=2$ \textit{slow} variables and $m=2$ \textit{fast} may be written as:

\begin{subequations}
\label{eq11}
\begin{align}
 \dot{x}_1 & = f_1 \left( x_1, x_2, y_1, y_2 \right), \hfill \\
 \dot{x}_2 & = f_2 \left( x_1, x_2, y_1, y_2 \right), \hfill \\
 \varepsilon \dot{y}_1 & = g_1 \left( x_1, x_2, y_1, y_2 \right), \hfill \\
 \varepsilon \dot{y}_2 & = g_2 \left( x_1, x_2, y_1, y_2 \right), \hfill
\end{align}
\end{subequations}

\smallskip

where $\vec{x}= (x_1, x_2)^t \in \mathbb{R}^2$, $\vec{y}= (y_1, y_2) \in \mathbb{R}^2$, $0 < \varepsilon \ll 1$ and the functions $f_i$ and $g_i$ are assumed to be $C^2$ functions of $(x_1, x_2, y_1, y_2)$.

\subsection{Critical Manifold}

The critical manifold equation of system (\ref{eq11}) is defined by setting $\varepsilon = 0$ in Eqs. (11c \& 11d). Thus, we obtain:

\begin{subequations}
\label{eq12}
\begin{align}
& g_1 \left( x_1, x_2, y_1, y_2 \right) = 0, \hfill \\
& g_2 \left( x_1, x_2, y_1, y_2 \right) = 0. \hfill
\end{align}
\end{subequations}

By application of the Implicit Function Theorem, let suppose that we can explicitly express from Eqs. (12a \& 12b), say without loss of generality, $x_1$ and $y_1$ as functions of the others variables:

\begin{subequations}
\label{eq13}
\begin{align}
& x_1 = \phi_1\left( x_2, y_1, y_2 \right), \hfill \\
& y_1 = \phi_2 \left( x_1, x_2, y_2 \right). \hfill
\end{align}
\end{subequations}

\subsection{Constrained system}

The \textit{constrained system} is obtained by equating to zero the time derivative of $g_{1,2} \left(x_1,x_2,y_1,y_2 \right)$:

\begin{subequations}
\label{eq14}
\begin{align}
\frac{dg_1}{dt} & = \frac{\partial g_1}{\partial x_1} \dot{x}_1 + \frac{\partial g_1}{\partial x_2} \dot{x}_2 + \frac{\partial g_1}{\partial y_1} \dot{y}_1 + \frac{\partial g_1}{\partial y_1} \dot{y}_2 = 0 \\
\frac{dg_2}{dt} & = \frac{\partial g_2}{\partial x_1} \dot{x}_1 + \frac{\partial g_2}{\partial x_2} \dot{x}_2 + \frac{\partial g_2}{\partial y_1} \dot{y}_1 + \frac{\partial g_2}{\partial y_1} \dot{y}_2 = 0
\end{align}
\end{subequations}

Eqs. (14a \& 14b) may be written as:

\begin{subequations}
\label{eq15}
\begin{align}
&  \frac{\partial g_1}{\partial y_1} \dot{y}_1 + \frac{\partial g_1}{\partial y_1} \dot{y}_2 = - \left( \frac{\partial g_1}{\partial x_1} \dot{x}_1 + \frac{\partial g_1}{\partial x_2} \dot{x}_2 \right) \\
&  \frac{\partial g_2}{\partial y_1} \dot{y}_1 + \frac{\partial g_2}{\partial y_1} \dot{y}_2 = - \left( \frac{\partial g_2}{\partial x_1} \dot{x}_1 + \frac{\partial g_2}{\partial x_2} \dot{x}_2  \right)
\end{align}
\end{subequations}

By solving the system of two equations (15a \& 15b) with two unknowns $(\dot{y}_1, \dot{y}_2)$ we find:

\begin{subequations}
\label{eq16}
\begin{align}
&  \dot{y}_1 = \dfrac{ - \left( \frac{\partial g_1}{\partial x_1} \dot{x}_1 + \frac{\partial g_1}{\partial x_2} \dot{x}_2 \right) \frac{\partial g_2}{\partial y_2} + \left( \frac{\partial g_2}{\partial x_1} \dot{x}_1 + \frac{\partial g_2}{\partial x_2} \dot{x}_2  \right) \frac{\partial g_1}{\partial y_2} }{det\left[ J_{(y_1,y_2)} \right]}, \\
&  \dot{y}_2 = \dfrac{ - \left( \frac{\partial g_1}{\partial x_1} \dot{x}_1 + \frac{\partial g_1}{\partial x_2} \dot{x}_2 \right) \frac{\partial g_2}{\partial y_1} + \left( \frac{\partial g_2}{\partial x_1} \dot{x}_1 + \frac{\partial g_2}{\partial x_2} \dot{x}_2  \right) \frac{\partial g_1}{\partial y_1} }{det\left[ J_{(y_1,y_2)} \right]}.
\end{align}
\end{subequations}

So, we have the following constrained system:

\begin{equation}
\label{eq17}
\begin{aligned}
 \dot{x}_1 & = f_1 \left( x_1, x_2, y_1, y_2 \right), \hfill \\
 \dot{x}_2 & = f_2 \left( x_1, x_2, y_1, y_2 \right), \hfill \\
 \dot{y}_1 & = \dfrac{ - \left( \frac{\partial g_1}{\partial x_1} \dot{x}_1 + \frac{\partial g_1}{\partial x_2} \dot{x}_2 \right) \frac{\partial g_2}{\partial y_2} + \left( \frac{\partial g_2}{\partial x_1} \dot{x}_1 + \frac{\partial g_2}{\partial x_2} \dot{x}_2  \right) \frac{\partial g_1}{\partial y_2} }{det\left[ J_{(y_1,y_2)} \right]}, \hfill \\
 \dot{y}_2 & = \dfrac{ - \left( \frac{\partial g_1}{\partial x_1} \dot{x}_1 + \frac{\partial g_1}{\partial x_2} \dot{x}_2 \right) \frac{\partial g_2}{\partial y_1} + \left( \frac{\partial g_2}{\partial x_1} \dot{x}_1 + \frac{\partial g_2}{\partial x_2} \dot{x}_2  \right) \frac{\partial g_1}{\partial y_1} }{det\left[ J_{(y_1,y_2)} \right]}, \hfill \\
 0 & = g_1 \left( x_1, x_2, y_1, y_2 \right), \hfill \\
0 & = g_2 \left( x_1, x_2, y_1, y_2 \right). \hfill
\end{aligned}
\end{equation}

\subsection{Normalized slow dynamics}

By rescaling the time by setting $t = - det\left[ J_{(y_1,y_2)} \right] \tau$ we obtain the ``normalized slow dynamics'':

\begin{equation}
\label{eq18}
\begin{aligned}
 \dot{x}_1 & = - f_1 \left( x_1, x_2, y_1, y_2 \right) det\left[ J_{(y_1,y_2)} \right] = F_1 \left( x_1, x_2, x_3, y_1 \right), \hfill \\
 \dot{x}_2 & = - f_2 \left( x_1, x_2, y_1, y_2 \right) det\left[ J_{(y_1,y_2)} \right] = F_2 \left( x_1, x_2, x_3, y_1 \right), \hfill \\
 \dot{y}_1 & = \left( \frac{\partial g_1}{\partial x_1} \dot{x}_1 + \frac{\partial g_1}{\partial x_2} \dot{x}_2 \right) \frac{\partial g_2}{\partial y_2} - \left( \frac{\partial g_2}{\partial x_1} \dot{x}_1 + \frac{\partial g_2}{\partial x_2} \dot{x}_2  \right) \frac{\partial g_1}{\partial y_2} \hfill \\
 & = G_1 \left( x_1, x_2, x_3, y_1 \right), \hfill \\
 \dot{y}_2 & = \left( \frac{\partial g_1}{\partial x_1} \dot{x}_1 + \frac{\partial g_1}{\partial x_2} \dot{x}_2 \right) \frac{\partial g_2}{\partial y_1} - \left( \frac{\partial g_2}{\partial x_1} \dot{x}_1 + \frac{\partial g_2}{\partial x_2} \dot{x}_2  \right) \frac{\partial g_1}{\partial y_1} \hfill \\
 & = G_2 \left( x_1, x_2, x_3, y_1 \right), \hfill \\
 0 & = g_1 \left( x_1, x_2, y_1, y_2 \right), \hfill \\
0 & = g_2 \left( x_1, x_2, y_1, y_2 \right). \hfill
\end{aligned}
\end{equation}

where the overdot now denotes the time derivation with respect to $\tau$.

\subsection{Desingularized system on the critical manifold}

Then, since we have supposed that $x_1$ and $y_1$ may be explicitly expressed as functions of the others variables (13a \& 13b), they can be used to project the normalized slow dynamics (\ref{eq18}) on the tangent bundle of the critical manifold. So, we have:

\begin{equation}
\label{eq19}
\begin{aligned}
 \dot{x}_2 & = - f_2 \left( x_1, x_2, y_1, y_2 \right) det\left[ J_{(y_1,y_2)} \right] = F_2 \left( x_2, y_2 \right), \hfill \\
 \dot{y}_2 & = ( \frac{\partial g_1}{\partial x_1} \dot{x}_1 + \frac{\partial g_1}{\partial x_2} \dot{x}_2 ) \frac{\partial g_2}{\partial y_1} - ( \frac{\partial g_2}{\partial x_1} \dot{x}_1 + \frac{\partial g_2}{\partial x_2} \dot{x}_2  ) \frac{\partial g_1}{\partial y_1} = G_2 \left( x_2, y_2 \right). \hfill
\end{aligned}
\end{equation}

\smallskip

\subsection{Pseudo singular points}

\textit{Pseudo-singular points} are defined as singular points of the ``normalized slow dynamics'', \textit{i.e.} as the set of points for which we have:

\begin{subequations}
\label{eq20}
\begin{align}
 & det\left[ J_{(y_1,y_2)} \right] = 0, \hfill \\
 & \left( \frac{\partial g_1}{\partial x_1} \dot{x}_1 + \frac{\partial g_1}{\partial x_2} \dot{x}_2 \right) \frac{\partial g_2}{\partial y_2} - \left( \frac{\partial g_2}{\partial x_1} \dot{x}_1 + \frac{\partial g_2}{\partial x_2} \dot{x}_2  \right) \frac{\partial g_1}{\partial y_2} = 0, \hfill \\
 & \left( \frac{\partial g_1}{\partial x_1} \dot{x}_1 + \frac{\partial g_1}{\partial x_2} \dot{x}_2 \right) \frac{\partial g_2}{\partial y_1} - \left( \frac{\partial g_2}{\partial x_1} \dot{x}_1 + \frac{\partial g_2}{\partial x_2} \dot{x}_2  \right) \frac{\partial g_1}{\partial y_1} = 0, \hfill \\
 & g_1 \left( x_1, x_2, y_1, y_2 \right) = 0, \hfill \\
 & g_2 \left( x_1, x_2, y_1, y_2 \right) = 0. \hfill
\end{align}
\end{subequations}

\begin{remark}
Let's notice on the one hand that Eqs. (20b) \& (20c) are linearly dependent and on the other hand that contrary to previous works we don't use the ``desingularized vector field'' (\ref{eq19}) but the ``normalized slow dynamics'' (\ref{eq18}).
\end{remark}

The Jacobian matrix of system (\ref{eq18}) reads:

\begin{equation}
\label{eq21}
J_{(F_1, F_2, G_1, G_2)} = \begin{pmatrix}
\dfrac{\partial F_1}{\partial x_1 } \quad &  \quad \dfrac{\partial F_1}{\partial x_2 } \quad &  \quad \dfrac{\partial F_1}{\partial y_1 } \quad &  \quad \dfrac{\partial F_1}{\partial y_2 } \vspace{6pt} \\
\dfrac{\partial F_2}{\partial x_1 } \quad &  \quad \dfrac{\partial F_2}{\partial x_2 } \quad &  \quad \dfrac{\partial F_2}{\partial y_1 } \quad &  \quad \dfrac{\partial F_2}{\partial y_2 } \vspace{6pt} \\
\dfrac{\partial G_1}{\partial x_1 } \quad &  \quad \dfrac{\partial G_1}{\partial x_2 } \quad &  \quad \dfrac{\partial G_1}{\partial y_1 } \quad &  \quad \dfrac{\partial G_1}{\partial y_2 } \vspace{6pt} \\
\dfrac{\partial G_2}{\partial x_1 } \quad &  \quad \dfrac{\partial G_2}{\partial x_2 } \quad &  \quad \dfrac{\partial G_2}{\partial y_1 } \quad &  \quad \dfrac{\partial G_2}{\partial y_2 } \vspace{6pt}
\end{pmatrix}
\end{equation}

\subsection{Extension of Beno\^{i}t's generic hypothesis}

Without loss of generality, it seems reasonable to extend Beno\^{i}t's generic hypotheses introduced for the three-dimensional case to the four-dimensional case.
So, first, let's suppose that by a ``standard translation'' the \textit{pseudo singular point} can be shifted at the origin $O(0,0,0,0)$ and that by a ``standard rotation'' of $y_1$-axis that the \textit{slow manifold} is tangent to ($x_2, x_3, y_1$)-hyperplane, so we have

\begin{equation}
\label{eq22}
\begin{aligned}
& f_1 \left(0, 0, 0, 0 \right) = g_1 \left( 0, 0, 0, 0 \right) = 0 \hfill \\
& \left.  \dfrac{\partial g_1}{\partial x_2} \right|_{(0,0,0,0)} = \left.  \dfrac{\partial g_1}{\partial x_3} \right|_{(0,0,0,0)} = \left.  \dfrac{\partial g_1}{\partial y_1} \right|_{(0,0,0, 0)} = 0
\end{aligned}
\end{equation}

\smallskip

Then, let's make the following assumptions for the non-degeneracy of the \textit{folded singularity}:

\begin{equation}
\label{eq23}
f_2 \left(0, 0, 0, 0 \right) \neq 0 \quad \mbox{ ; } \quad  \left.  \dfrac{\partial g_1}{\partial x_1} \right|_{(0,0,0,0)} \neq 0 \quad \mbox{ ; } \quad \left.  \dfrac{\partial^2 g_1}{\partial y_1^2} \right|_{(0,0,0,0)} \neq 0.
\end{equation}

According to these generic hypotheses Eqs. (\ref{eq22}-\ref{eq23}), the Jacobian matrix (\ref{eq21}) reads:

\begin{equation}
\label{eq24}
J_{(F_1, F_2, G_1, G_2)} = \begin{pmatrix}
0 &  0  &  0  &  0  \vspace{6pt} \\
-f_2 \dfrac{\partial P}{\partial x_1}  &  -f_2 \dfrac{\partial P}{\partial x_2}   &  -f_2 \dfrac{\partial P}{\partial y_1}   &   - f_2 \dfrac{\partial P}{\partial y_2} \vspace{6pt} \\
a_{31}  &  a_{32}   &  a_{33}   &   a_{34} \vspace{6pt} \\
a_{41}  &  a_{42}   &   a_{43}  &  a_{44} \vspace{6pt}
\end{pmatrix}
\end{equation}

where

\[
\begin{aligned}
P & = det\left[ J_{(y_1,y_2)} \right], \hfill \\
a_{3i} & = -f_2 \dfrac{\partial g_2}{\partial x_2} \dfrac{\partial^2 g_1}{\partial y_2 \partial x_i} + \dfrac{\partial g_2}{\partial y_2} \left( f_2 \dfrac{\partial^2 g_1}{\partial x_2 \partial x_i} + \dfrac{\partial g_1}{\partial x_1} \dfrac{\partial f_1}{\partial x_i} \right) \mbox{ for } i =1,2,\\
a_{3i} & = -f_2 \dfrac{\partial g_2}{\partial x_2} \dfrac{\partial^2 g_1}{\partial y_2 \partial y_i} + \dfrac{\partial g_2}{\partial y_2} \left( f_2 \dfrac{\partial^2 g_1}{\partial x_2 \partial y_i} + \dfrac{\partial g_1}{\partial x_1} \dfrac{\partial f_1}{\partial y_i} \right) \mbox{ for } i =3,4,\\
a_{4i} & = f_2 \dfrac{\partial g_2}{\partial x_2} \dfrac{\partial^2 g_1}{\partial y_1 \partial x_i} - \dfrac{\partial g_2}{\partial y_1} \left( f_2 \dfrac{\partial^2 g_1}{\partial x_2 \partial x_i} + \dfrac{\partial g_1}{\partial x_1} \dfrac{\partial f_1}{\partial x_i} \right) \mbox{ for } i =1,2,\\
a_{4i} & = f_2 \dfrac{\partial g_2}{\partial x_2} \dfrac{\partial^2 g_1}{\partial y_1 \partial y_i} - \dfrac{\partial g_2}{\partial y_1} \left( f_2 \dfrac{\partial^2 g_1}{\partial x_2 \partial y_i} + \dfrac{\partial g_1}{\partial x_1} \dfrac{\partial f_1}{\partial y_i} \right) \mbox{ for } i =3,4.\\
\end{aligned}
\]

Thus, we have the following Cayley-Hamilton eigenpolynomial associated with such a Jacobian matrix (\ref{eq24}) evaluated at the \textit{pseudo singular point}, \textit{i.e.}, at the origin:

\begin{equation}
\label{eq25}
\lambda^4 - \sigma_1 \lambda^3 + \sigma_2 \lambda^2 - \sigma_3 \lambda + \sigma_4 = 0
\end{equation}

where $\sigma_1 = Tr(J)$ is the sum of all first-order diagonal minors of $J$, \textit{i.e.}, the the trace of the Jacobian matrix $J$, $\sigma_2$ represents the sum of all second-order diagonal minors of $J$ and $\sigma_3$ represents the sum of all third-order diagonal minors of $J$. It appears that $\sigma_4 = |J|=0$ since one row of the Jacobian matrix (\ref{eq24}) is null. So, the eigenpolynomial reduces to:

\begin{equation}
\label{eq26}
\lambda\left( \lambda^3 - \sigma_1 \lambda^2 + \sigma_2 \lambda - \sigma_3 \right) = 0
\end{equation}

But, according to Wechselberger \cite{Wechselberger2012}, $\sigma_3$ vanishes at a \textit{pseudo singular point} as it's easy to prove it. So, the eigenpolynomial (\ref{eq26}) is reduced to

\begin{equation}
\label{eq27}
\lambda^2\left( \lambda^2 - \sigma_1 \lambda + \sigma_2 \right) = 0
\end{equation}

Let $\lambda_i$ be the eigenvalues of the eigenpolynomial (\ref{eq27}) and let's denote by $\lambda_{3,4} = 0$ the obvious double root of this polynomial. We have:

\begin{equation}
\label{eq28}
\begin{aligned}
\sigma_1 & = Tr(J_{(F_1, F_2, G_1, G_2)} ) = \lambda_1 + \lambda_2 = \dfrac{\partial g_2}{\partial x_1}\dfrac{\partial g_1}{\partial y_1}\dfrac{\partial f_1}{\partial y_2}, \hfill  \\
\sigma_2 & = \sum_{i=1}^3 \left| J_{(F_1, F_2, G_1, G_2)}^{ii} \right| = \lambda_1\lambda_2 \hfill \\
& = \left( \dfrac{\partial g_1}{\partial y_1}\right)^2 \left[ f_2^2 \left( \dfrac{\partial^2 g_2}{\partial x_2^2 } \dfrac{\partial^2 g_2}{\partial y_2^2} -  \left( \dfrac{\partial^2 g_2}{\partial x_2 \partial y_2} \right)^2 \right) \right. \hfill \\
& \left. + f_2 \dfrac{\partial g_2}{\partial x_1}\left( \dfrac{\partial^2 g_2}{\partial y_2^2 } \dfrac{\partial f_1}{\partial x_2} - \dfrac{\partial^2 g_2}{\partial x_2 \partial y_2}\dfrac{\partial f_1}{\partial y_2} \right) \right]
\end{aligned}
\end{equation}

where $\sigma_1 = Tr(J_{(F_1, F_2, G_1, G_2)} ) = p$ is is the sum of all first-order diagonal minors of $J_{(F_1, F_2, G_1, G_2)}$, \textit{i.e.} the trace of the Jacobian matrix $J_{(F_1, F_2, G_1, G_2)} $ and $\sigma_2 = \sum_{i=1}^3 \left| J_{(F_1, F_2, G_1, G_2)}^{ii} \right| = q$ represents the sum of all second-order diagonal minors of $J_{(F_1, F_2, G_1, G_2)} $. Thus, the \textit{pseudo singular point} is of saddle-type \textit{iff} the following conditions $C_1$ and $C_2$ are verified:

\begin{equation}
\label{eq29}
\begin{aligned}
C_1:& \quad \Delta = p^2 - 4q > 0, \hfill \\
C_2:& \quad q < 0.
\end{aligned}
\end{equation}

Condition $C_1$ is systematically satisfied provided that condition $C_2$ is verified. Thus, the \textit{pseudo singular point} is of saddle-type \textit{iff} $q < 0$.

\subsection{Canard existence in $\mathbb{R}^{2+2}$}

Following the works of Wechselberger \cite{Wechselberger2012} it can be stated, while using a standard polynomial change of variables, that any $n$-dimensional singularly perturbed systems with $k$ \textit{slow} variables ($k \geqslant 2$) and $m$ \textit{fast} ($m \geqslant 1$) (\ref{eq1}) can be transformed into the following ``normal form'':

\begin{equation}
\label{eq30}
\begin{aligned}
\dot{x_1} & = \tilde{a} x_2 + \tilde{b} y_2 + O \left( x_1, \epsilon, x_2^2, x_2 y_2, y_2^2 \right), \hfill \\
\dot{x_2} & = 1 + O \left( x_1, x_2, y_2, \epsilon \right), \hfill \\
\epsilon \dot{y_1} & = \tilde{c} y_1 + O \left(\epsilon x_1, \epsilon x_2, \epsilon y_2, x_1^2, x_2^2, y_2^2, x_2 y_2  \right), \hfill \\
\epsilon \dot{y_2} & =  -\left( x_1 + y_2^2 \right) + O \left( \epsilon x_1, \epsilon x_2, \epsilon y_2, \epsilon^2, x_1^2 y_2, y_2^3, x_1 x_2 y_2  \right). \hfill
\end{aligned}
\end{equation}

We establish in Appendix A for any four-dimensional singularly perturbed systems (\ref{eq11}) with $k=2$ \textit{slow} and $m=2$ \textit{fast} variables that

\[
\begin{aligned}
\tilde{a} & = \frac{1}{2} \left[ f_2^2 \left( \dfrac{\partial^2 g_2}{\partial x_2^2 } \dfrac{\partial^2 g_2}{\partial y_2^2} - \left( \dfrac{\partial^2 g_2}{\partial x_2 \partial y_2} \right)^2 \right) + f_2 \dfrac{\partial g_2}{\partial x_1} \left( \dfrac{\partial^2 g_2}{\partial y_2^2 }\dfrac{\partial f_1}{\partial x_2} - \dfrac{\partial^2 g_2}{\partial x_2 \partial y_2 } \dfrac{\partial f_1}{\partial y_2} \right) \right] \hfill \\
\tilde{b} & = - \dfrac{\partial g_2}{\partial x_1}\dfrac{\partial f_1}{\partial y_2}, \hfill  \\
\tilde{c} & = \dfrac{\partial g_1}{\partial y_1}. \hfill
\end{aligned}
\]

Thus, in his paper Wechselberger \cite[p. 3304]{Wechselberger2012} provided in the framework of ``standard analysis'' a generalization of Beno\^{i}t's theorem \cite{Benoit1983} for any $n$-dimensional singularly perturbed systems with $k$ \textit{slow} variables ($k \geqslant 2$) and $m$ \textit{fast} ($m \geqslant 1$). According to his Theorem 4.1 presented below he proved the existence of canard solutions for the original system (\ref{eq1}).

\begin{theorem}\hfill \\
\label{theo2}
In the folded saddle case of system {\rm (30)} singular canards perturb to maximal canards solutions for sufficiently small $\varepsilon \ll 1$.
\end{theorem}

\begin{proof}
See Wechselberger \cite{Wechselberger2012}.
\end{proof}

Since our method doesn't use the ``desingularized vector field'' (\ref{eq19}) but the ``normalized slow dynamics'' (\ref{eq18}), we have the following proposition:

\begin{proposition}\hfill \\
\label{prop1}
If the normalized slow dynamics {\rm (\ref{eq18})} has a \textit{pseudo singular point} of saddle type, \textit{i.e.} if the sum $\sigma_2$ of all second-order diagonal minors of the Jacobian matrix of the normalized slow dynamics {\rm (\ref{eq18})} evaluated at the \textit{pseudo singular point} is negative, \textit{i.e.} if $\sigma_2 < 0$ then, according to Theorem \ref{theo2}, system {\rm (\ref{eq11})} exhibits a canard solution which evolves from the attractive part of the slow manifold towards its repelling part.

\end{proposition}

\begin{proof}

By making some smooth changes of time and smooth changes of coordinates (see Appendix A) we brought the system (\ref{eq11}) to the following ``normal form'':

\[
\begin{aligned}
\dot{x_1} & = \tilde{a} x_2 + \tilde{b} y_2 + O \left( x_1, \epsilon, x_2^2, x_2 y_2, y_2^2 \right), \hfill \vspace{6pt} \\
\dot{x_2} & = 1 + O \left( x_1, x_2, y_2, \epsilon \right), \hfill \vspace{6pt} \\
\epsilon \dot{y_1} & = \tilde{c} y_1 + O \left(\epsilon x_1, \epsilon x_2, \epsilon y_2, x_1^2, x_2^2, y_2^2, x_2 y_2  \right), \hfill \vspace{6pt} \\
\epsilon \dot{y_2} & =  -\left( x_1 + y_2^2 \right) + O \left( \epsilon x_1, \epsilon x_2, \epsilon y_2, \epsilon^2, x_1^2 y_2, y_2^3, x_1 x_2 y_2  \right), \hfill
\end{aligned}
\]

Then, we deduce that the condition for the \textit{pseudo singular point} to be of saddle type is $\tilde{a} < 0$. According to Eqs. (\ref{eq29}) it is easy to verify that

\[
\begin{aligned}
\sigma_1 & = Tr(J_{(F_1, F_2, G_1, G_2)}) = \lambda_1 + \lambda_2 = -\tilde{b}\tilde{c}, \hfill  \\
\sigma_2 & = \sum_{i=1}^3 \left| J_{(F_1, F_2, G_1, G_2)}^{ii} \right| = \lambda_1\lambda_2 = 2\tilde{a}\tilde{c}^2.
\end{aligned}
\]

\smallskip

So, the condition for which the \textit{pseudo singular point} is of saddle type, \textit{i.e.} $\sigma_2 < 0$ is identical to that proposed by Wechselberger \cite[p. 3298]{Wechselberger2012} in his theorem, \textit{i.e.} $\tilde{a} < 0$.
\end{proof}

So, Prop. \ref{prop1} can be used to state the existence of canard solution for such systems. Application of Proposition \ref{prop1} to the coupled FitzHugh-Nagumo equations, presented in Sec. 4, which is a four-dimensional singularly perturbed system with two \textit{slow} and two \textit{fast} variables will enable to prove, as many previous works such as those of Tchizawa \& Campbell \cite{Tchizawa2002} and Tchizawa \cite{Tchizawa2002,Tchizawa2004,Tchizawa2007,Tchizawa2008,Tchizawa2010,Tchizawa2012}, the existence of ``canard solutions'' in such system. According to Tchizawa \cite{Tchizawa2014}, it is very important to notice, on the one hand that the fast equation has 2-dimensional in the system $\mathbb{R}^{2+2}$ and, on the other hand that the fast system can give attractive, repulsive or attractive-repulsive at each \textit{pseudo singular point}. Then, Tchizawa \cite{Tchizawa2014} has established that the jumping direction can be shown using the eigenvectors. In the same way we will find again the results of Rubin \textit{et al.} \cite{Rubin} concerning the existence of ``canard solutions'' in the Hodgkin-Huxley model but with a set of more realistic parameters used in Chua \textit{et al.} \cite{ChuaSbitnev2012a,ChuaSbitnev2012b}.

\newpage

\section{Coupled FitzHugh-Nagumo equations}
\label{Sec4}

The FitzHugh-Nagumo model \cite{Fitzhugh,Nagumo} is a simplified version of the Hodgkin-Huxley model \cite{Hodgkin} which models in a detailed manner activation and deactivation dynamics of a spiking neuron. By coupling two FitzHugh-Nagumo models Tchizawa \& Campbell \cite{TchiCamp} and Tchizawa \cite{Tchizawa2002,Tchizawa2012} obtained the following four-dimensional singularly perturbed system with two \textit{slow} and two \textit{fast} variables:

\begin{subequations}
\label{eq31}
\begin{align}
\dfrac{dx_1}{dt} & = \dfrac{1}{c} \left( y_1 + b x_1 \right), \hfill \\
\dfrac{dx_2}{dt} & = \dfrac{1}{c} \left( y_2 + b x_2 \right), \hfill \\
\varepsilon \dfrac{dy_1}{dt} & = x_1 - \frac{y_1^3}{3} + y_2, \hfill \\
\varepsilon \dfrac{dy_2}{dt} & = x_2 - \frac{y_2^3}{3} + y_1.
\end{align}
\end{subequations}

\smallskip

where $0 < \varepsilon \ll 1$ and $b$ is the ``canard parameter'' or ``duck parameter'' while $c$ is a scale factor.

\subsection{Slow manifold and contrained system}

The slow manifold equation of system (\ref{eq31}) is defined by setting $\varepsilon = 0$ in Eqs. (31c \& 31d). Thus, we obtain:

\begin{equation}
\label{eq32}
\begin{aligned}
\dfrac{dx_1}{dt} & = \dfrac{1}{c} \left( y_1 + b x_1 \right), \hfill \vspace{6pt} \\
\dfrac{dx_2}{dt} & = \dfrac{1}{c} \left( y_2 + b x_2 \right), \hfill \vspace{6pt} \\
\dfrac{dy_1}{dt} & = - \dfrac{ \dfrac{1}{c} \left( y_2 + b x_2 \right) + \dfrac{y_2^2}{c} \left( y_1 + b x_1 \right)  }{  y_1^2 y_2^2 -1 }, \hfill \vspace{6pt} \\
\dfrac{dy_2}{dt} & = - \dfrac{ \dfrac{1}{c} \left( y_1 + b x_1 \right) + \dfrac{y_1^2}{c} \left( y_2 + b x_2 \right)  }{  y_1^2 y_2^2 -1 }, \hfill \vspace{6pt} \\
0 & = x_1 - \frac{y_1^3}{3} + y_2, \vspace{6pt} \\
0 & = x_2 - \frac{y_2^3}{3} + y_1.
\end{aligned}
\end{equation}

\subsection{Normalized slow dynamics}

Then, by rescaling the time by setting $t = - det\left[ J_{(y_1,y_2)} \right] \tau = - (y_1^2 y_2^2 -1) $ we obtain the ``normalized slow dynamics'':

\begin{equation}
\label{eq33}
\begin{aligned}
\dfrac{dx_1}{dt} & = - \dfrac{1}{c} \left( y_1 + b x_1 \right) \left( y_1^2 y_2^2 - 1 \right)= F_1 \left( x_1, x_2, y_1, y_2 \right), \hfill \vspace{6pt} \\
\dfrac{dx_2}{dt} & = - \dfrac{1}{c} \left( y_2 + b x_2 \right) \left( y_1^2 y_2^2 - 1 \right) = F_2 \left( x_1, x_2, y_1, y_2 \right), \hfill \vspace{6pt} \\
\dfrac{dy_1}{dt} & = \dfrac{1}{c} \left( y_2 + b x_2 \right) + \dfrac{y_2^2}{c} \left( y_1 + b x_1 \right) = G_1 \left( x_1, x_2, y_1, y_2 \right), \hfill \vspace{6pt} \\
\dfrac{dy_2}{dt} & = \dfrac{1}{c} \left( y_1 + b x_1 \right) + \dfrac{y_1^2}{c} \left( y_2 + b x_2 \right) G_2 \left( x_1, x_2, y_1, y_2 \right), \hfill \vspace{6pt} \\
0 & = x_1 - \frac{y_1^3}{3} + y_2, \vspace{6pt} \\
0 & = x_2 - \frac{y_2^3}{3} + y_1.
\end{aligned}
\end{equation}

\subsection{Pseudo singular points}

From Eqs. (\ref{eq20}), the \textit{pseudo-singular points} of system (\ref{eq31}) are defined by:

\begin{subequations}
\label{eq34}
\begin{align}
& det\left[ J_{(y_1,y_2)} \right] = y_1^2 y_2^2 - 1 = 0, \hfill \\
& \left( \frac{\partial g_1}{\partial x_1} \dot{x}_1 + \frac{\partial g_1}{\partial x_2} \dot{x}_2 \right) \frac{\partial g_2}{\partial y_2} - \left( \frac{\partial g_2}{\partial x_1} \dot{x}_1 + \frac{\partial g_2}{\partial x_2} \dot{x}_2  \right) \frac{\partial g_1}{\partial y_2} \hfill \\
& = \dfrac{1}{c} \left( y_2 + b x_2 \right) + \dfrac{y_2^2}{c} \left( y_1 + b x_1 \right) = 0, \hfill  \nonumber \\
& \left( \frac{\partial g_1}{\partial x_1} \dot{x}_1 + \frac{\partial g_1}{\partial x_2} \dot{x}_2 \right) \frac{\partial g_2}{\partial y_1} - \left( \frac{\partial g_2}{\partial x_1} \dot{x}_1 + \frac{\partial g_2}{\partial x_2} \dot{x}_2  \right) \frac{\partial g_1}{\partial y_1} \hfill \\
& = \dfrac{1}{c} \left( y_1 + b x_1 \right) + \dfrac{y_1^2}{c} \left( y_2 + b x_2 \right) = 0, \hfill \nonumber \\
& g_1 \left( x_1, x_2, y_1, y_2 \right) = x_1 - \frac{y_1^3}{3} + y_2 = 0, \hfill \\
& g_2 \left( x_1, x_2, y_1, y_2 \right) = x_2 - \frac{y_2^3}{3} + y_1 = 0.
\end{align}
\end{subequations}

According to Tchizawa \& Campbell \cite{TchiCamp} and Tchizawa \cite{Tchizawa2002,Tchizawa2004}, there are six \textit{pseudo singular points}, the last four are depending on the parameter $b$.

\begin{subequations}
\label{eq35}
\begin{align}
\left( \tilde{x}_{1}, \tilde{x}_{2}, \tilde{y}_{1}, \tilde{y}_{2} \right) = & \left( \pm \dfrac{4}{3}, \mp \dfrac{4}{3}, \pm 1, \mp 1 \right), \hfill \\
\left( \tilde{x}_{1}, \tilde{x}_{2}, \tilde{y}_{1}, \tilde{y}_{2} \right) = & \left( \pm \frac{\sqrt{\frac{3-\sqrt{9-4 b^2}}{b}} \left(3+2 \sqrt{9-4 b^2}\right)}{3 \sqrt{2} b} \right., \nonumber \\
& \mp \frac{\sqrt{\frac{3-\sqrt{9-4 b^2}}{b}} \left(9-8 b^2+3 \sqrt{9-4 b^2}\right)}{6 \sqrt{2} b^2}, \nonumber \\
& \mp\sqrt{\frac{3-\sqrt{9-4 b^2}}{2b}}, \left. \mp \frac{\sqrt{2b}}{\sqrt{3-\sqrt{9-4 b^2}}} \right), \hfill \\
\left( \tilde{x}_{1}, \tilde{x}_{2}, \tilde{y}_{1}, \tilde{y}_{2} \right) = & \left( \pm \frac{\left(3-2 \sqrt{9-4 b^2}\right) \sqrt{\frac{3+\sqrt{9-4 b^2}}{b}}}{3 \sqrt{2} b} \right., \nonumber  \\
& \mp \frac{\sqrt{\frac{3+\sqrt{9-4 b^2}}{b}} \left(9 -8 b^2 - 3 \sqrt{9-4 b^2}\right)}{6 \sqrt{2} b^2}, \nonumber\\
& \mp\sqrt{\frac{3-\sqrt{9-4 b^2}}{2b}}, \left. \mp \frac{\sqrt{2b}}{\sqrt{3-\sqrt{9-4 b^2}}} \right).
\end{align}
\end{subequations}

\subsection{Canard existence in coupled FitzHugh-Nagumo equations}

The Jacobian matrix of system (\ref{eq33}) evaluated at the \textit{pseudo singular points} (\ref{eq35}a) reads:

\begin{equation}
\label{eq36}
J_{(F_1, F_2, G_1, G_2)} = \begin{pmatrix}
0 \quad &  0 \quad & \quad &  \dfrac{2 (3+4 b)}{3 c} \quad &  -\dfrac{2 (3+4 b)}{3 c} \vspace{6pt} \\
0 \quad &  0 \quad & \quad &  -\dfrac{2 (3+4 b)}{3 c} \quad &  \dfrac{2 (3+4 b)}{3 c} \vspace{6pt} \\
\dfrac{b}{c} \quad &  \dfrac{b}{c} \quad & \quad &  \dfrac{1}{c} \quad &  -\frac{3+8 b}{3 c}  \vspace{6pt} \\
\dfrac{b}{c} \quad &  \dfrac{b}{c} \quad & \quad &  \quad -\dfrac{3+8 b}{3 c} \quad &  \quad \dfrac{1}{c}  \vspace{6pt}
\end{pmatrix}
\end{equation}

\begin{remark}
Although the \textit{pseudo singular points} have not been shifted at the origin extension of Beno\^{i}t's generic hypotheses (\ref{eq22}-\ref{eq23}) are satisfied. In other words, we have $\sigma_4 = \sigma_3 = 0$.
\end{remark}

According to Eqs. (\ref{eq28}) we find that:

\begin{equation}
\label{eq37}
\begin{aligned}
p & = \sigma_1 = Tr(J) = +\dfrac{2}{c}, \hfill  \\
q & = \sigma_2 = -\dfrac{16 b (3+4 b)}{9 c^2}
\end{aligned}
\end{equation}

\smallskip

Thus, according to Prop. \ref{prop1}, the \textit{pseudo singular points} are of saddle-type if and only if:

\[
 -\dfrac{16 b (3+4 b)}{9 c^2} < 0
\]

So, we have the following conditions $C_1$ and $C_2$:

\begin{equation}
\label{eq38}
\begin{aligned}
C_1:& \quad \Delta = \frac{4 (3 + 8 b)^2}{9 c^2} > 0, \hfill \\
C_2:& \quad q = - \dfrac{16 b (3 + 4 b)}{9 c^2}  < 0.
\end{aligned}
\end{equation}

\smallskip

Let's choose arbitrarily $b$ as the ``canard parameter'' or ``duck parameter''. Obviously, it appears that the condition $C_1$ is still satisfied. Finally, the \textit{pseudo singular points} are of saddle-type if and only if we have:

\begin{equation}
\label{eq39}
b > 0 \quad \mbox{ or } \quad b < - \dfrac{3}{4}.
\end{equation}

\smallskip

\begin{remark}
Let's notice that the \textit{pseudo singular points} are of node-type if $-\dfrac{3}{4} < b < 0$ as stated by Tchizawa \& Campbell \cite{TchiCamp} and Tchizawa \cite{Tchizawa2002,Tchizawa2004}.
\end{remark}

The Jacobian matrix $J_{(F_1, F_2, G_1, G_2)}$ of system (\ref{eq33}) evaluated at the \textit{pseudo singular points} (\ref{eq35}b) reads:

\begin{equation}
\label{eq40}
\begin{pmatrix}
0  &  0  &  &  -\frac{4 \sqrt{9-4 b^2}}{3 c}  &  -\frac{2 \left(-9+4 b^2+3 \sqrt{9-4 b^2}\right)}{3 b c} \vspace{6pt} \\
0  &  0  &  &  \frac{2 \left(-9+4 b^2+3 \sqrt{9-4 b^2}\right)}{3 b c}  &  \frac{4 \sqrt{9-4 b^2}}{3 c} \vspace{6pt} \\
\frac{3+\sqrt{9-4 b^2}}{2 c}  &  \frac{b}{c}  &  &  \frac{3+\sqrt{9-4 b^2}}{2 b c}  &  \frac{3-4 \sqrt{9-4 b^2}}{3 c}  \vspace{6pt} \\
\frac{b}{c}  &  \frac{3-\sqrt{9-4 b^2}}{2 c}  &  &   \frac{3 + 4 \sqrt{9-4 b^2}}{3 c}  &   \frac{3 - \sqrt{9-4 b^2}}{2 b c}  \vspace{6pt}
\end{pmatrix}
\end{equation}

\begin{remark}
Although, the \textit{pseudo singular points} have not been shifted at the origin extension of Beno\^{i}t's generic hypotheses (\ref{eq22}-\ref{eq23}) are satisfied. In other words, we have $\sigma_4 = \sigma_3 = 0$.
\end{remark}

According to Eqs. (\ref{eq28}) we find that:

\begin{equation}
\label{eq41}
\begin{aligned}
p & = \sigma_1 = Tr(J) = + \dfrac{3}{b c}, \hfill  \\
q & = \sigma_2 = \dfrac{16 \left(9-4 b^2\right)}{9 c^2}
\end{aligned}
\end{equation}

\smallskip

Thus, according to Prop. \ref{prop1}, the \textit{pseudo singular points} are of saddle-type if and only if:

\[
\dfrac{16 \left(9-4 b^2\right)}{9 c^2} < 0
\]

\[
\Delta = p^2 - 4q > 0 \qquad \mbox{ and } \qquad q <0.
\]

So, we have the following conditions $C_1$ and $C_2$:

\begin{equation}
\label{eq42}
\begin{aligned}
C_1:& \quad \Delta = \left( \dfrac{3}{bc} \right)^2 - \dfrac{64 \left(9-4 b^2\right)}{9 c^2}  > 0, \hfill \\
C_2:& \quad q = \dfrac{16 \left(9-4 b^2\right)}{9 c^2}  < 0.
\end{aligned}
\end{equation}

\smallskip

Let's choose arbitrarily $b$ as the ``canard parameter'' or ``duck parameter''. Obviously, it appears that if the condition $C_2$ is verified then the condition $C_1$ is \textit{de facto} satisfied. Finally, the \textit{pseudo singular points} are of saddle-type if and only if we have:

\begin{equation}
\label{eq43}
b > \dfrac{3}{2} \quad \mbox{ or } \quad b < - \dfrac{3}{2}.
\end{equation}

\smallskip

\begin{remark}
Because of the symmetry of this coupled FitzHugh-Nagumo equations, the Jacobian matrix of system (\ref{eq33}) evaluated at the \textit{pseudo singular points} (\ref{eq35}c) provides the same result as just above.
\end{remark}

\newpage

\section{Hodgkin-Huxley model}
\label{Sec5}

The original Hodgkin-Huxley model \cite{Hodgkin} is described by the following system of four nonlinear ordinary differential
equations:

\begin{subequations}
\label{eq44}
\begin{align}
\frac{dV}{dt} & = \frac{1}{C_M}\left[I -\bar{g}_Kn^4(V-V_K) - \bar{g}_{Na}m^3h(V-V_{Na}) - \bar{g}_L (V-V_L)\right]\\
\frac{dn}{dt} & =  \alpha_n(V) (1-n) - \beta_n(V)n\\
\frac{dm}{dt} & = \alpha_m(V) (1-m) - \beta_m(V)m\\
\frac{dh}{dt} & = \alpha_h(V) (1-h) - \beta_h(V)h
\end{align}
\end{subequations}

where:

\begin{subequations}
\begin{align}
\label{eq45}
\alpha_n(V) & = 0.01(V + 10) / \left( \exp \dfrac{V + 10}{10} - 1 \right),\\
\beta_n(V) & = 0.125 \exp \left( V/80 \right), \\
\alpha_m(V) & = 0.1(V + 25) / \left( \exp \dfrac{V + 25}{10} - 1 \right), \\
\beta_m(V) & = 4 \exp (V / 18),\\
\alpha_h(V) & = 0.07 \exp (V / 20), \\
\beta_n(V) & = 1 / \left( \exp \frac{V + 30}{10} + 1 \right)
\end{align}
\end{subequations}

The first equation (\ref{eq44}a) results from the application of Kirchhoff's law to the space clamped squid giant axon. Thus, the total membrane current $C_M dV / dt$ for which $C_M$ represents the specific membrane capacity and $V$ the displacement of the membrane potential from its resting value, is equal to the sum of the following intrinsic currents:

\begin{eqnarray*}
I_K  & = & \bar{g}_Kn^4(V-V_K) \\
I_{Na} & = & \bar{g}_{Na}m^3h(V-V_{Na})\\
I_L & = & \bar{g}_L (V-V_L)
\end{eqnarray*}

where $I_K$ is a delayed rectifier potassium current, $I_{Na}$ is fast sodium current and $I_L$ is the ``leakage current''. The parameter $I$ is the total membrane current density, inward positive, \textit{i.e.} the total current injected into the space clamped squid giant axon and $V_K$, $V_{Na}$ and $V_L$ are the equilibrium potentials of potassium, sodium and ``leakage current'' respectively. The maximal specific conductances of the ionic currents are denoted $\bar{g}_K$, $\bar{g}_{Na}$ and $\bar{g}_L$ respectively. Functions $\alpha_{n,m,h}$ and $\beta_{n,m,h}$ are gates' opening and closing rates depending on $V$. Variable $m$ denotes the activation of the sodium current, variable $h$ the inactivation of the sodium current and variable $n$ the activation of the potassium current. These dimensionless gating variables vary between $[0,1]$.

Let's notice that the variables and symbols in Eqs. (44 \& 45) originally chosen by Hodgkin-Huxley and are different from those found in recent literatures where the reference polarity of the voltage $V$, and the reference direction of the current $I$ are defined as the negative of the voltages and currents. We have opted to adopt the reference assumption in Hodgkin \& Huxley \cite{Hodgkin} for ease in
comparison of our results with those from Hodgkin and Huxley\footnote{For more details see Chua \textit{et al.} \cite{ChuaSbitnev2012a,ChuaSbitnev2012b}}.
The parameter values are exactly those chosen in the original Hodgkin-Huxley \cite{Hodgkin} works:

\begin{eqnarray*}
C_M & = &  1.0 \mbox{ }\mu F/cm^2\\
\vspace{0.1pt}
V_{Na} & = & -115 \mbox{ } mV\\
\vspace{0.1pt}
V_K & = & 12 \mbox{ } mV\\
\vspace{0.1pt}
V_L & = & -10.613 \mbox{ } mV\\
\vspace{0.1pt}
\bar{g}_{Na} & = & 120 \mbox{ } mS/cm^2\\
\vspace{0.1pt}
\bar{g}_K & = & 36 \mbox{ } mS/cm^2\\
\vspace{0.1pt}
\bar{g}_L & = & 0.3 \mbox{ } mS/cm^2\\
\end{eqnarray*}

According to Suckley and Biktashev \cite{SuckBikt} and Suckley \cite{Suckley}, dimensionless functions $\bar{n}$, $\bar{h}$ and $\bar{m}$ called gates' instant equilibrium values, \textit{i.e.}, steady-state relation for gating variable $n$, $h$ and $m$ respectively as well as $\tau_n$, $\tau_h$ and $\tau_m$ called gates dynamics time scales in $ms$, \textit{i.e.}, time constant for gating variable $n$, $h$ and $m$ respectively may be defined as follows:

\begin{subequations}
\label{eq46}
\begin{align}
\bar{n}(V) & = \frac{\alpha_n (V)}{\alpha_n (V) + \beta_n (V) } \\
\bar{h}(V) & = \frac{\alpha_h (V)}{\alpha_h (V) + \beta_h (V) } \\
\bar{m}(V) & = \frac{\alpha_m (V)}{\alpha_m (V) + \beta_m (V) }\\
\tau_n (V) & = \frac{1}{\alpha_n (V) + \beta_n (V)} \\
\tau_h (V) & = \frac{1}{\alpha_h (V) + \beta_h (V)} \\
\tau_m (V) & = \frac{1}{\alpha_m (V) + \beta_m (V)}
\end{align}
\end{subequations}

\vspace{0.1in}

By using Eqs. \ref{eq46}, the original Hodgkin-Huxley model \cite{Hodgkin} reads:

\begin{subequations}
\label{eq47}
\begin{align}
\frac{dV}{dt} & = \frac{1}{C_M}\left[I -\bar{g}_Kn^4(V-V_K) - \bar{g}_{Na}m^3h(V-V_{Na}) - \bar{g}_L (V-V_L)\right]\\
\frac{dn}{dt} & = \frac{\bar{n} - n}{\tau_n}\\
\frac{dh}{dt} & = \frac{\bar{h} - h}{\tau_h}\\
\frac{dm}{dt} & = \frac{\bar{m} - m}{\tau_m}
\end{align}
\end{subequations}

Now, in order to apply the \textit{singular perturbation method} to the Hodgkin-Huxley model, two small multiplicative parameters $\varepsilon \ll 1$ are introduced. According to  Suckley and Biktashev \cite{SuckBikt}, Suckley \cite{Suckley} and Rubin and Wechselberger \cite{Rubin}, the existence of two different time scales of evolution for the couples of dynamic variables ($n,h$) and ($m,V$) enables to justify such an introduction. So, in order to differentiate \textit{slow} variables from \textit{fast} variables, Suckley and Biktashev \cite{SuckBikt}, Suckley \cite{Suckley} and Rubin and Wechselberger \cite{Rubin} have plotted the inverse of ``time constant for gating variable $i$'', \textit{i.e.}, ${\tau_i}^{-1}$ according to $V$ with $i=n,h,m$. In Fig. 1, they have been plotted for the original functions ${\alpha_i}$ and ${\beta_i}$ (Eqs. \ref{eq45}). However, let's notice that this plot is exactly the same as those presented by Rubin and Wechselberger \cite{Rubin} (Fig. 1) for a nondimensionalized three-dimensional Hodgkin-Huxley singularly perturbed system obtained after the following variable changes: $V \to -V$ and $\bar{I} \to -\bar{I}$, then $V \to V + 65$ and finally $V \to V / 100$.

\begin{figure}[htbp]
\centerline{\includegraphics[width=10cm,height=10cm]{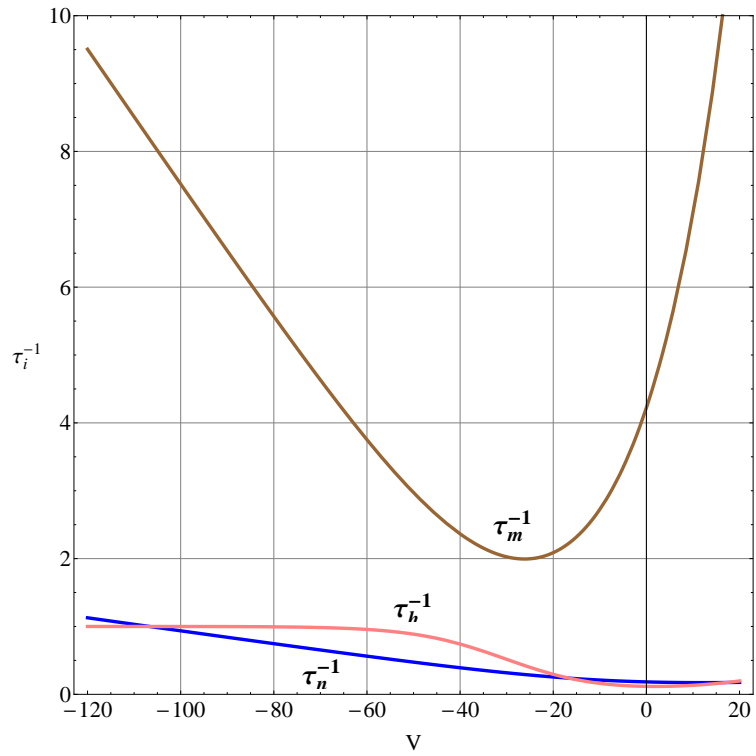}}
\caption{Graph of $1 / \tau_i$ (ms$^{-1}$) against V ($mV$).}
\label{Fig3}
\end{figure}

Fig. 1 shows a plot of the functions ${\tau_i}^{-1}$ according to $V$ with $i=n,h,m$ over the physiological range. We observe that ${\tau_m}^{-1}$ is of an order of magnitude bigger than ${\tau_h}^{-1}$ and ${\tau_n}^{-1}$, which are of comparable size. Indeed, we can deduce that the values of times scales are approximately ${\tau_m}^{-1} \approx 10 ms^{-1}$ while ${\tau_n}^{-1} \approx {\tau_h}^{-1} \approx 1 ms^{-1}$. Then, it appears that $m$ corresponds to the \textit{fast} variable while $n$ and $h$ correspond to \textit{slow} variables. Moreover, since the activation of the sodium channel $m$ is directly related to the dynamics of the membrane (action) potential $V$, Rubin and Wechselberger \cite{Rubin} consider that $m$ and $V$ evolve on the same \textit{fast} time scale. So, the Hodgkin-Huxley model may be transformed into a \textit{singularly perturbed system} with two time scales in which the \textit{slow} variables are ($n,h$) and the \textit{fast} variables are ($m,V$).\\

So, according to Awiszus \textit{et al.} \cite{Awisus}, Suckley and Biktashev \cite{SuckBikt}, Suckley \cite{Suckley} and Rubin and Wechselberger \cite{Rubin} small multiplicative parameters $0 < \varepsilon \ll 1$ in the original vector field of the Hodgkin-Huxley Eqs. (\ref{eq47}) may be identified while factorizing the right hand side of Eq. (\ref{eq47}a) by $\bar{g}_{Na}$ and set:

\begin{eqnarray*}
\bar{g}_{Na} & \to & \frac{\bar{g}_{Na}}{\bar{g}_{Na}} = 1 \mbox{, } \bar{g}_{K} \to \frac{\bar{g}_{K}}{\bar{g}_{Na}} = 0.3 \mbox{, } \bar{g}_{L} \to \frac{\bar{g}_{L}}{\bar{g}_{Na}} = 0.0025.
\end{eqnarray*}

other parameters are kept as for the original Hodgkin-Huxley model \cite{Hodgkin}:

\begin{eqnarray*}
C_M & = &  1.0 \mbox{ }\mu F/cm^2 \mbox{, } \bar{V}_{Na} =  -115 \mbox{ } mV \mbox{, } \bar{V}_K = 12 \mbox{ } mV \mbox{, } \bar{V}_L = -10.613 \mbox{ } mV.
\end{eqnarray*}

Then, by posing $\bar{I} \to \dfrac{ \bar{I} }{\bar{g}_{Na}}$, $\varepsilon = \dfrac{ C_M }{\bar{g}_{Na} } = \dfrac{1}{120}$ and ($n, h, m, V$) $=$ ($x_1, x_2, y_1, y_2$) to consistent with the notations of Sec. 3, we obtain:

\begin{subequations}
\label{eq48}
\begin{align}
\frac{dx_1}{dt} & = \frac{\bar{x}_1 - x_1}{\tau_1}  = f_1 \left( x_1, x_2, y_1, y_2 \right) \\
\frac{dx_2}{dt} & = \frac{\bar{x}_2 - x_2}{\tau_2} = f_1 \left( x_1, x_2, y_1, y_2 \right) \\
\varepsilon  \frac{dy_1}{dt} & = \frac{\bar{y}_1 - y_1}{\tau_3} = g_1 \left( x_1, x_2, y_1, y_2 \right)\\
\varepsilon \frac{dy_2}{dt} & = \bar{I} -\bar{g}_K x^4_1 (y_2 - V_K) - \bar{g}_{Na} y^3_1 x_2 (y_2 - V_{Na}) - \bar{g}_L (y_2 - V_L) \nonumber \\
&  =  g_2 \left( x_1, x_2, y_1, y_2 \right)
\end{align}
\end{subequations}

\smallskip

where ($\bar{x}_1, \bar{x}_2, \bar{y}_1$)  $=$ ($\bar{n}, \bar{h}, \bar{m}$) and $\tau_{1,2,3} = \tau_{n,h,m}$.\\

Let's notice that the multiplicative parameter $\varepsilon$ has been introduced artificially in Eq. (\ref{eq48}c). This is due to the fact that it has been stated above that the time scale of variable $m$, \textit{i.e.}, $y_1$ is tenth times greater than the time scale of variables $n$ and $h$, \textit{i.e.} of variables $x_1$ and $x_2$. Moreover, this parameter is identical to those use in Eq. (\ref{eq48}d) since it has been also considered that $m$ and $V$, \textit{i.e.}, $y_1$ and $y_2$ evolve on the same \textit{fast} time scale.\\

According to the \textit{Geometric Singular Perturbation Theory}, the zero-order approximation in $\varepsilon$ of the \textit{slow manifold} associated with the Hodgkin-Huxley model (48) is obtained by posing $\varepsilon = 0$ in Eqs. (\ref{eq48}c \& \ref{eq48}d). So, the \textit{slow manifold} is given by:

\begin{subequations}
\label{eq49}
\begin{align}
x_2 & = \frac{\bar{I} -\bar{g}_K x^4_1 (y_2 - V_K) - \bar{g}_L (y_2 - V_L)}{\bar{g}_{Na} \bar{y}^3_1 (y_2 - V_{Na})}\\
y_1 & = \bar{y}_1 (y_2)
\end{align}
\end{subequations}

\vspace{0.1in}

Then, the \textit{fast foliation} is within the planes $x_1 = constant$ and $x_2 = constant$.\\

The \textit{fold curve} is defined as the location of the points where $g_1\left(x_1,x_2,y_1,y_2 \right) = 0$, $g_2\left(x_1,x_2,y_1,y_2 \right) = 0$ and $det\left[ J_{(g_1, g_2)} \right] = 0$. For the Hodgkin-Huxley model (\ref{eq48}), the \textit{fold curve} is thus given by Eqs. (\ref{eq49}a \& \ref{eq49}b) and by the determinant of the Jacobian matrix of the following \textit{fast foliation}:

\begin{subequations}
\label{eq50}
\begin{align}
\frac{dy_1}{dt} & = \frac{\bar{y}_1 - y_1}{\tau_3} = g_1 \left( x_1, x_2, y_1, y_2 \right)\\
\frac{dy_2}{dt} & = \bar{I} -\bar{g}_K x^4_1 (y_2 - V_K) - \bar{g}_{Na} y^3_1 x_2 (y_2 - V_{Na}) - \bar{g}_L (y_2 - V_L) \nonumber \\
 & = g_2 \left( x_1, x_2, y_1, y_2 \right)
\end{align}
\end{subequations}

\smallskip

The Jacobian matrix of the \textit{fast foliation} (\ref{eq50}) reads:

\begin{equation}
\label{eq51}
J_{(g_1, g_2)} = \begin{pmatrix}
\dfrac{\bar{y_1}'\tau_3 - \tau'_3(\bar{y_1} - y_1)}{\tau^2_3} \quad & \quad -\dfrac{1}{\tau_3} \vspace{6pt} \\
-(\bar{g}_K x^4 + \bar{g}_{Na} y^3_1 x_2 + \bar{g}_L ) \quad & \quad  - 3 \bar{g}_{Na} y^2_1 x_2 (y_2 - V_{Na})\\
\end{pmatrix}
\end{equation}

\vspace{0.1in}

where the ($'$) denotes the derivative with respect to $y_2$. Then, taking into account Eqs. (\ref{eq49}b), \textit{i.e.}, $y_1 = \bar{y}_1$ we have:

\begin{equation}
\label{eq52}
J_{(g_1, g_2)} = \begin{pmatrix}
\dfrac{\bar{y_1}'}{\tau_3} \quad & \quad -\dfrac{1}{\tau_3} \vspace{6pt} \\
-(\bar{g}_K x^4_1 + \bar{g}_{Na} \bar{y}^3_1 x_2 + \bar{g}_L ) \quad & \quad  - 3 \bar{g}_{Na} \bar{y}^2_1 x_2 (y_2 - V_{Na})\\
\end{pmatrix}
\end{equation}

\vspace{0.1in}

So, the determinant of the Jacobian matrix of the \textit{fast foliation} (\ref{eq50}) is:

\begin{equation}
\label{eq53}
det\left(J_{(g_1, g_2)} \right) = - \frac{1}{\tau_3} \left[ \bar{g}_K x^4_1 +  \bar{g}_{Na} \bar{y}^3_1 x_2 + \bar{g}_L + 3 \bar{g}_{Na} \bar{y_1}' \bar{y}^2_1 x_2 (y_2 - V_{Na}) \right]
\end{equation}

\vspace{0.1in}

Thus, the condition for the \textit{fold curve} is $det\left(J_{(g_1, g_2)} \right) = 0$, which gives:

\begin{equation}
\label{eq54}
\bar{g}_K x^4_1 +  \bar{g}_{Na} \bar{y}^3_1 x_2 + \bar{g}_L + 3 \bar{g}_{Na} \bar{y_1}' \bar{y}^2_1 x_2 (y_2 - V_{Na}) = 0
\end{equation}

Therefore:

\begin{equation}
\label{eq55}
x_2 = - \frac{\bar{g}_K x^4_1 + \bar{g}_L}{\bar{g}_{Na} \bar{y}^2_1 \left( \bar{y}_1 + 3 \bar{y_1}' (y_2 - V_{Na}) \right)}
\end{equation}

\vspace{0.1in}

By subtracting Eq. (\ref{eq49}a) from Eq. (\ref{eq55}) we obtain $x_1$:

\begin{equation}
\label{eq56}
x_1 = x_\textit{1f} =  \left[\frac{-\bar{I}\left[ \bar{y}_1 + 3 \bar{y_1}' (y_2 - V_{Na}) \right] + \bar{g}_L (V_{Na} - V_L)\bar{y}_1 + 3 \bar{y_1}' (y_2 - V_{Na})(y_2 - V_L)}{\bar{g}_K \left[ (V_K - V_{Na})\bar{y}_1 - 3 \bar{y_1}' (y_2 - V_{Na})(y_2 - V_K) \right] } \right]^{1/4}
\end{equation}

\vspace{0.1in}

Plugging this value of $x_1$ (\ref{eq56}) into Eq. (\ref{eq55}) provides:

\begin{equation}
\label{eq57}
x_2 = x_\textit{2f} =\frac{\bar{I} + \bar{g}_L (V_K - V_L)}{\bar{g}_{Na} \bar{y}^2_1 \left[ (V_{Na} - V_K )\bar{y}_1 + 3 \bar{y_1}' (y_2 - V_{Na})(y_2 - V_K) \right]}
\end{equation}

\vspace{0.1in}

So, the \textit{fold curve} is given by the set of parametric equations (\ref{eq56}-\ref{eq57}) in terms of $y_2$.\\

The \textit{pseudo singular points} are given by Eqs. (\ref{eq20}) which reads for the Hodgkin-Huxley model (\ref{eq48}):

\begin{subequations}
\label{eq58}
\begin{align}
& \frac{\bar{y}_1 - y_1}{\tau_3} = 0, \\
& \bar{I} -\bar{g}_K x^4_1 (y_2 - V_K) - \bar{g}_{Na} y^3_1 x_2 (y_2 - V_{Na}) - \bar{g}_L (y_2 - V_L) = 0, \\
& \left[ \frac{ 4 \bar{g}_K x^3_1 (y_2 - V_K) (x_1 - \bar{x}_1) }{ \tau_1 } + \frac{ \bar{g}_{Na} y^3_1 (y_2 - V_{Na}) (x_2 - \bar{x}_2)  }{ \tau_2 } \right] = 0, \\
&  \left[ \frac{ 4 \bar{g}_K x^3_1 (y_2 - V_K) (x_1 - \bar{x}_1) }{ \tau_1 } + \frac{ \bar{g}_{Na} y^3_1 (y_2 - V_{Na}) (x_2 - \bar{x}_2)  }{ \tau_2 } \right]\frac{1}{\tau_3}  = 0, \\
& \tau_3 \left( \bar{g}_K x^4_1 + \bar{g}_{Na} y^3_1 x_2 + \bar{g}_L  \right)  + 3 \bar{g}_{Na} y^2_1 x_2 (y_2 - V_{Na}) (\tau_3 y_1' + (y_1 - \bar{y}_1 ) \tau_3')  = 0.
\end{align}
\end{subequations}

\vspace{0.1in}

Let's notice that Eqs. (\ref{eq58}c) and (\ref{eq58}d) are identical. Moreover, the definition of $\tau_3$ (\ref{eq46}f) enables to simplify the above system (\ref{eq58}). Thus, we have:

\begin{subequations}
\label{eq59}
\begin{align}
& \bar{I} -\bar{g}_K x^4_1 (y_2 - V_K) - \bar{g}_{Na} \bar{y}^3_1 x_2 (y_2 - V_{Na}) - \bar{g}_L (y_2 - V_L) = 0, \\
&  \frac{ 4 \bar{g}_K x^3_1 (y_2 - V_K) (x_1 - \bar{x}_1) }{ \tau_1 } + \frac{ \bar{g}_{Na} \bar{y}^3_1 (y_2 - V_{Na}) (x_2 - \bar{x}_2)  }{ \tau_2 }= 0, \\
& \bar{g}_K x^4_1 + \bar{g}_{Na} \bar{y}^3_1 x_2 + \bar{g}_L  + 3 \bar{g}_{Na} \bar{y}^2_1 \bar{y}_1' x_2 (y_2 - V_{Na})  = 0.
\end{align}
\end{subequations}

\smallskip

Moreover, Eqs. (\ref{eq59}a) and (\ref{eq59}c) indicate that the \textit{pseudo singular point} belongs to the \textit{slow manifold} and to the \textit{fold curve}. So, let's replace in Eq. (\ref{eq59}b) the variables $x_1$ and $x_2$ by the variables $x_\textit{1f}$ and $x_\textit{2f}$ given by Eq. (\ref{eq56}) and Eq. (\ref{eq57}) respectively which represent the parametric equations of \textit{fold curve}.

\begin{equation}
\label{eq60}
\frac{ 4 g_K x^3_\textit{1f} (y_2 - V_K) (x_\textit{1f} - \bar{x}_1) }{ \tau_1 } + \frac{ \bar{y}^3_1 (y_2 - V_{Na}) (x_\textit{2f} - \bar{x}_2)  }{ \tau_2 }  = 0.
\end{equation}

\vspace{0.1in}

Thus, it appears that Eq. (\ref{eq60}) depends on the variable $y_2$, on the functions gates dynamics time scales $\tau_1(y_2)$ and $\tau_2(y_2)$ and on the bifurcation parameter $\bar{I}$. According to Rubin and Wechselberger \cite{Rubin}, the function $y_2(\bar{I})$, solution of (\ref{eq60}) is independent of time multiplicative constants $k_1$ and $k_2$ that one could set in factor of $\tau_1(y_2)$ and $\tau_2(y_2)$.\\

So, following their works, let's plot the function $y_2(\bar{I})$ solution of (\ref{eq60}) for various values of these time constants by posing successively in (\ref{eq60}) $k_1 = 1$, $3$, $4.75$ and $7$ and while fixing $k_2 =1$. The result is presented in Fig. 2.

\begin{figure}[htbp]
\centerline{\includegraphics[width=10cm,height=10cm]{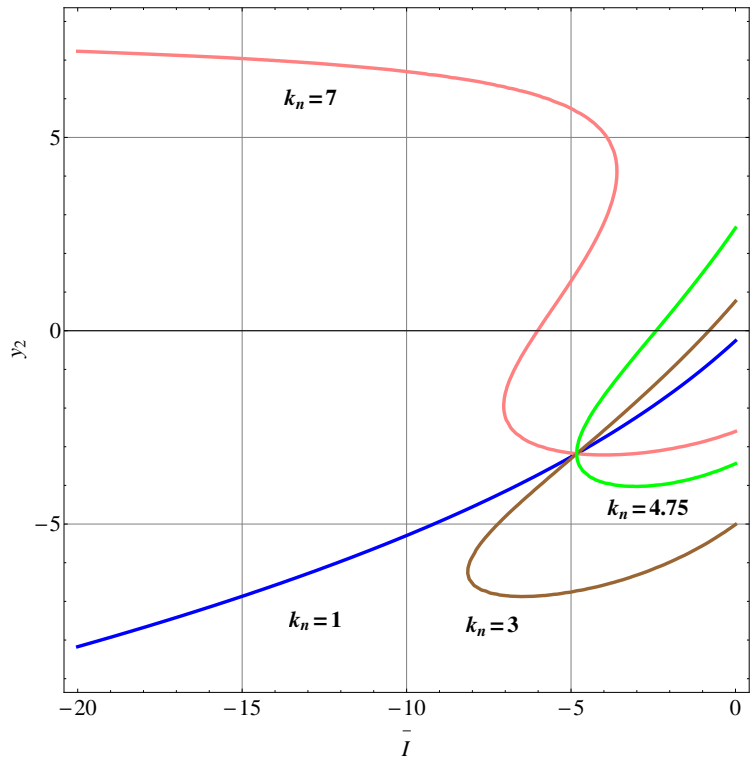}}
\caption{Function $y_2(\bar{I})$ for various values of parameter $k_n = 1, 3, 4.75, 7$ \\ \centering{\vphantom{NN}\hspace{1cm}exhibiting the the bifurcation parameter value $\bar{I}_C \approx - 4.8$}.}
\label{Fig4}
\end{figure}

\smallskip

Let's notice that this plot\footnote{The function $y_2(\bar{I})$ solution of (\ref{eq59}) has been plotted with Mathematica$^\copyright$ while using the ContourPlot function used for representing implicit function since such function cannot be expressed explicitly.} is exactly the same as those presented by Rubin and Wechselberger \cite{Rubin} (Fig. 8-9) for a nondimensionalized three-dimensional Hodgkin-Huxley singularly perturbed system which had been obtained after the following variable changes: $V \to -V$ and $\bar{I} \to -I$, then $V \to V + 65$ and finally $V \to V / 100$.

We observe from Fig. 2 that the bifurcation parameter value $\bar{I}_C \approx - 4.8$ is exactly identical (in absolute value) to those obtained by Rubin and Wechselberger \cite{Rubin}. Numerical resolution\footnote{This resolution has been made while using the function FindRoot in Mathematica$^\copyright$.} of Eq. (\ref{eq60}) provides a better approximation of the bifurcation parameter value:

\[
\bar{I}_C = -4.82988 \mbox{ } \mu A
\]

\smallskip

This value corresponds to a voltage $y_2 = -3.18136 \mbox{ } mV$.\\

For $\bar{I} \approx - 4.1$, the coordinate of the \textit{pseudo singular point} can be computed numerically:

\[
(x_1,x_2,y_1,y_2) = (0.362513, 0.521793, 0.0733782, -2.81908)
\]

\smallskip

According to Proposition \ref{prop1} we can state that the eigenpolynomial of the Jacobian matrix associated with the ``normalized slow dynamics'' of the Hodgkin-Huxley model (\ref{eq48}) reads:

\[
\lambda^4 - \sigma_1 \lambda^3 + \sigma_2 \lambda^2 - \sigma_3 \lambda + \sigma_4 = 0
\]

for which it is easy to prove that $\sigma_4 = \sigma_3 = 0$. So, this eigenpolynomial reduces to:

\[
\lambda^2\left( \lambda^2 - \sigma_1 \lambda + \sigma_2 \right) = 0
\]

\smallskip

According to Eqs. (\ref{eq28}) we find that:

\[
\begin{aligned}
p & = Tr(J) = 144.933, \hfill  \\
q & = \sigma_2 = -362.924
\end{aligned}
\]

\smallskip

Thus, according to Prop. \ref{prop1}, the \textit{pseudo singular points} is of saddle-type. Moreover, numerical computation of the eigenvalues of this Jacobian matrix evaluated at the \textit{pseudo singular point} provides:

\[
(\lambda_1, \lambda_2, \lambda_3, \lambda_4) = (-2.46224, 147.396, 0,0)
\]

\vspace{0.1in}

So, according to Proposition \ref{prop1}, this \textit{pseudo singular point} is of saddle-type and canard solution may occur in the four-dimensional Hodgkin-Huxley singularly perturbed system (\ref{eq48}) for the original set of parameter values.\\

In Fig. 3, 4 \& 5 canard solution of the four-dimensional Hodgkin-Huxley singularly perturbed system for the ``canard value'' of $\bar{I} \approx -4.1$ has been plotted in the ($x_1,x_2,y_2$) phase-space and then in the ($x_1,y_2$) phase plane. The green point represents the \textit{pseudo singular point}. The trajectory curve, \textit{i.e.}, the canard solution has been plotted in red while the \textit{fold curve} is in yellow. We observe on Fig. 3 that when the trajectory curve reaches the \textit{fold} at the \textit{pseudo singular point} it ``jump'' suddenly to the other part of the \textit{slow manifold} before being reinjected towards the \textit{pseudo singular point}.

\begin{figure}[htbp]
\centerline{\includegraphics[width=10cm,height=10cm]{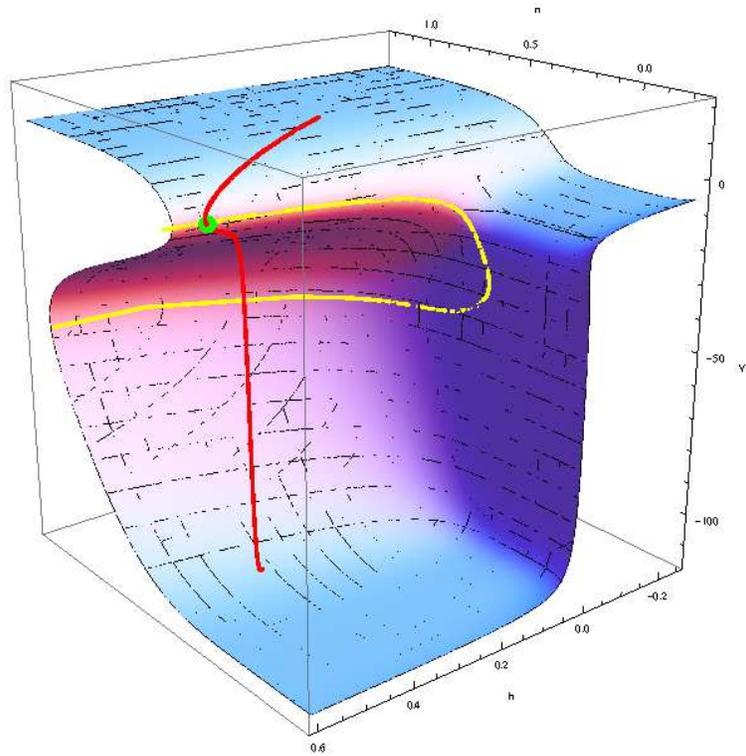}}
\caption{Phase portrait, canard solution and \textit{slow manifold} of the Hodgkin-Huxley system (48) in the ($n,h,V$) phase space.}
\label{Fig3}
\end{figure}

\begin{figure}[htbp]
\centerline{\includegraphics[width=10cm,height=10cm]{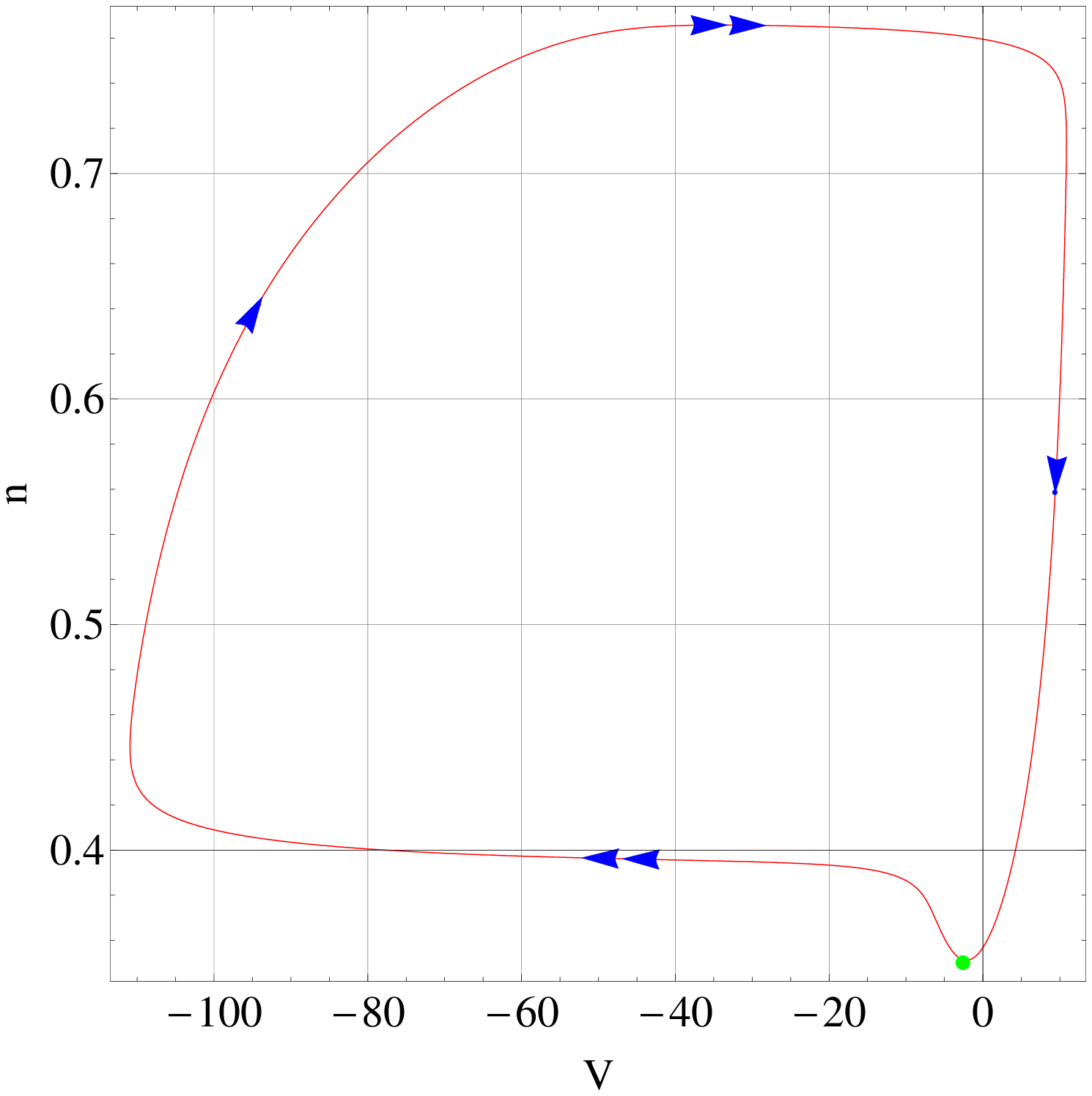}}
\caption{Phase portrait, canard solution and \textit{slow manifold} of the Hodgkin-Huxley system (48) in the ($V,n$) phase plane.}
\label{Fig4}
\end{figure}

\begin{figure}[htbp]
\centerline{\includegraphics[width=10cm,height=10cm]{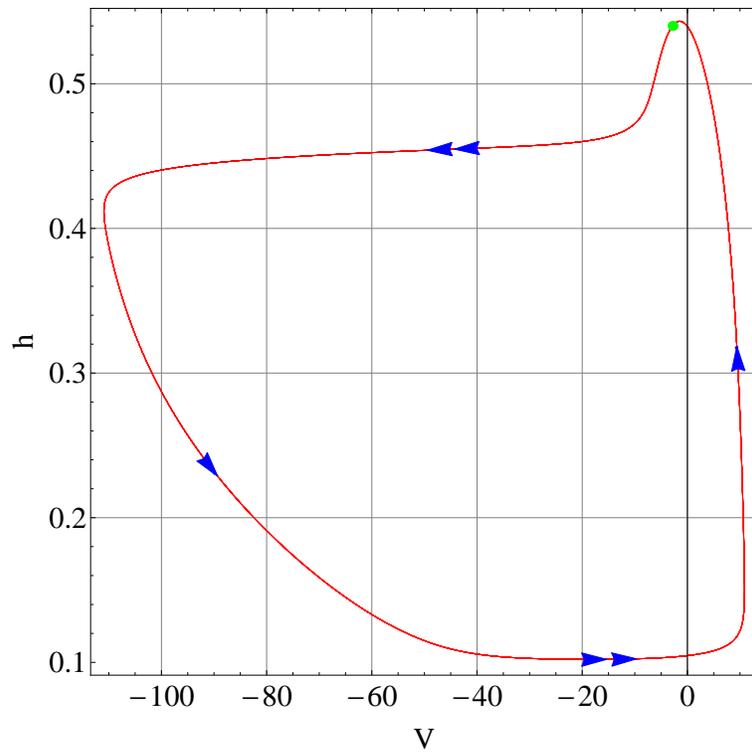}}
\caption{Phase portrait, canard solution and \textit{slow manifold} of the Hodgkin-Huxley system (48) in the ($V,h$) phase plane.}
\label{Fig5}
\end{figure}

\newpage

\section{Discussion}

In a previous paper entitled: ``Canards Existence in Memristor's Circuits'' (see Ginoux \& Llibre \cite{GinouxLLibre2015}) we have proposed a new method for proving the existence of ``canard solutions'' for three and four-dimensional singularly perturbed systems with only one \textit{fast} variable which improves the methods used until now. This method enabled to state a unique ``generic'' condition for the existence of ``canard solutions'' for such three and four-dimensional singularly perturbed systems which is based on the stability of \textit{folded singularities} of the \textit{normalized slow dynamics} deduced from a well-known property of linear algebra. This unique condition which is completely identical to that provided by Beno\^{i}t \cite{Benoit1983} and then by Szmolyan and Wechselberger \cite{SzmolyanWechselberger2001} and finally by Wechselberger \cite{Wechselberger2012} was considered as ``generic'' since it was exactly the same for singularly perturbed systems of dimension three and four with only one \textit{fast} variable. In this work we have extended this new method to the case of four-dimensional singularly perturbed systems with two \textit{slow} and two \textit{fast} variables and we have stated that the condition for the existence of ``canard solutions'' in such systems is exactly identical to those proposed in our previous paper. This result confirms the genericity of the condition ($\sigma_2 < 0$) we have highlighted and provides a simple and efficient tool for testing the occurrence of ``canard solutions'' in any three or four-dimensional singularly perturbed systems with one or two \textit{fast} variables. Applications of this method to the famous coupled FitzHugh-Nagumo equations and to the Hodgkin-Huxley model has enabled to show the existence of ``canard solutions'' in such systems. However, in this paper, only the case of \textit{pseudo singular points} or \textit{folded singularities} of saddle-type has been analyzed. Of course, the case of of \textit{pseudo singular points} or \textit{folded singularities} of node-type and focus-type could be also studied with the same method.

\section{Acknowledgements}
We would like to thank to Ernesto P\'erez Chavela  for previous
discussions related with this work. The authors are partially supported by a MINECO/FEDER grant number
MTM2008-03437. The second author is partially supported by a
MICINN/FEDER grants numbers MTM2009-03437 and MTM2013-40998-P, by an
AGAUR grant number 2014SGR-568, by an ICREA Academia, two
FP7+PEOPLE+2012+IRSES numbers 316338 and 318999, and
FEDER-UNAB10-4E-378.

\renewcommand{\theequation}{A-\arabic{equation}}
\setcounter{equation}{0}  

\newpage

\section*{Appendix}

Change of coordinates leading to the \textit{normal forms} of four-dimensional singularly perturbed systems with two fast variables are given in the following section.

\subsection*{Normal form of 4D singularly perturbed systems\\ with two fast variables}

Let's consider the four-dimensional \textit{singularly perturbed dynamical system} (\ref{eq11}) with $k=2$ \textit{slow} variables and $m=2$ \textit{fast} and let's make the following change of variables:

\begin{equation}
x_1 = \alpha^2 x, \quad  x_2 = \alpha y, \quad  y_1 = \alpha^2 z, \quad y_2 = \alpha u \quad \mbox{where} \quad \alpha \ll 1.
\end{equation}

By taking into account extension of Beno\^{i}t's generic hypothesis Eqs. (\ref{eq22},\ref{eq23}) and while using Taylor series expansion the system (\ref{eq11}) becomes:

\begin{equation}
\label{eqA2}
\begin{aligned}
\dot{x} & = \dfrac{\partial f_1}{\partial y} y + \dfrac{\partial f_1}{\partial u} u, \hfill \vspace{6pt} \\
\dot{y} & = f_2 \left( x, y, z, u \right), \hfill \vspace{6pt} \\
\left( \dfrac{\varepsilon}{\alpha} \right) \dot{z} & = \dfrac{\partial g_1}{\partial z} z + \dfrac{1}{2} \dfrac{\partial^2 g_1}{\partial y^2} y^2 + \dfrac{1}{2} \dfrac{\partial^2 g_1}{\partial u^2} u^2 + \dfrac{\partial^2 g_1}{\partial y \partial u} y u, \hfill \vspace{6pt} \\
\left( \dfrac{\varepsilon}{\alpha} \right) \dot{u} & = \dfrac{\partial g_2}{\partial x} x + \dfrac{1}{2} \dfrac{\partial^2 g_2}{\partial y^2} y^2 + \dfrac{1}{2} \dfrac{\partial^2 g_2}{\partial u^2} u^2 +  \dfrac{\partial^2 g_2}{\partial y \partial u} y u. \hfill
\end{aligned}
\end{equation}

\vspace{0.1in}

Then, let's make the standard polynomial change of variables:

\begin{equation}
\label{eqA3}
\begin{aligned}
X &  = A x + B y^2, \hfill \vspace{6pt} \\
Y & = \dfrac{y}{f_2}, \hfill \vspace{6pt} \\
Z & = C y + D z + E u, \hfill \vspace{6pt} \\
U & = F y + G u. \hfill
\end{aligned}
\end{equation}

\smallskip

From (A-3) we deduce that:

\begin{equation}
\label{eqA4}
\begin{aligned}
x & = \frac{X - B f_2^2 Y^2}{A}, \hfill \vspace{6pt} \\
y & = f_2 Y, \hfill \vspace{6pt} \\
z & = \frac{1}{D} \left[ Z - C f_2 Y - \frac{E}{G} \left( U - F f_2 Y \right) \right], \hfill \vspace{6pt} \\
u & = \frac{U - F f_2 Y }{G}. \hfill
\end{aligned}
\end{equation}

The time derivative of system (A-3) gives:

\begin{equation}
\label{eqA5}
\begin{aligned}
\dot{X} & = A \dot{x} + 2 B y \dot{y}, \hfill \vspace{6pt} \\
\dot{Y} & = \dfrac{\dot{y}}{f_2}, \hfill \vspace{6pt} \\
\dot{Z} & = C \dot{y} + D \dot{z} + E \dot{u}, \hfill \vspace{6pt} \\
\dot{U} & =  F\dot{z} + G \dot{u}. \hfill
\end{aligned}
\end{equation}

Then, multiplying the third and fourth equation of (A-5) by $( \varepsilon / \alpha)$ and while replacing in (A-5) $\dot{x}$, $\dot{y}$, $\dot{z}$ and $\dot{u}$ by the right-hand-side of system (A-2) leads to:

\begin{equation}
\label{eqA6}
\begin{aligned}
\dot{X} & = A \dot{x} + 2 B y \dot{y}, \hfill \vspace{6pt} \\
\dot{Y} & = \dfrac{\dot{y}}{f_2}, \hfill \vspace{6pt} \\
\left( \dfrac{\varepsilon}{\alpha} \right)\dot{Z} & = \left( \dfrac{\varepsilon}{\alpha} \right) C \dot{y} + \left( \dfrac{\varepsilon}{\alpha} \right) D \dot{z} + \left( \dfrac{\varepsilon}{\alpha} \right) E \dot{u}, \hfill \vspace{6pt} \\
\left( \dfrac{\varepsilon}{\alpha} \right)\dot{U} & = \left( \dfrac{\varepsilon}{\alpha} \right) F \dot{y} + \left( \dfrac{\varepsilon}{\alpha} \right) G \dot{u}. \hfill
\end{aligned}
\end{equation}

\smallskip

Since $ \varepsilon / \alpha \ll 1$, the first terms of the right-hand-side of the third and fourth equation of (A-16) can be neglected. So we have:

\begin{equation}
\label{eqA7}
\begin{aligned}
\dot{X} & = A \left( \dfrac{\partial f_1}{\partial y} y  + \dfrac{\partial f_1}{\partial u} u\right)  + 2 B f_2 y, \hfill \vspace{6pt} \\
\dot{Y} & = 1, \hfill \vspace{6pt} \\
\left( \dfrac{\varepsilon}{\alpha} \right)\dot{Z} & = D \left( \dfrac{\partial g_1}{\partial z} z + \dfrac{1}{2}\dfrac{\partial^2 g_1}{\partial y^2} y^2 + \dfrac{1}{2}\dfrac{\partial^2 g_1}{\partial u^2} u^2 + \dfrac{\partial^2 g_1}{\partial y \partial u} y u \right) \hfill \\
 & + E \left( \dfrac{\partial g_2}{\partial x} x + \dfrac{1}{2}\dfrac{\partial^2 g_2}{\partial y^2} y^2 + \dfrac{1}{2}\dfrac{\partial^2 g_2}{\partial u^2} u^2 + \dfrac{\partial^2 g_2}{\partial y \partial u} y u \right), \hfill \vspace{6pt} \\
\left( \dfrac{\varepsilon}{\alpha} \right)\dot{U} & = G \left( \dfrac{\partial g_2}{\partial x} x + \dfrac{1}{2}\dfrac{\partial^2 g_2}{\partial y^2} y^2 + \dfrac{1}{2}\dfrac{\partial^2 g_2}{\partial u^2} u^2 + \dfrac{\partial^2 g_2}{\partial y \partial u} y u \right). \hfill
\end{aligned}
\end{equation}

\smallskip

Then, by replacing in (A-7) $x$, $y$, $z$ and $u$ by the right-hand-side of (A-4) and by identifying with the following system in which we have posed: $(\varepsilon / \alpha) = \epsilon$:

\begin{equation}
\label{eqA8}
\begin{aligned}
\dot{X} & = \tilde{a} Y + \tilde{b} U + O \left( X, \epsilon, Y^2, Y U, U^2 \right), \hfill \vspace{6pt} \\
\dot{Y} & = 1 + O \left( X, Y, U, \epsilon \right), \hfill \vspace{6pt} \\
\epsilon \dot{Z} & = \tilde{c} Z + O \left( \epsilon X, \epsilon Y, \epsilon U, X^2, y^2 U, U^2, Y U  \right), \hfill \vspace{6pt} \\
\epsilon \dot{U} & =  -\left( X + U^2 \right) + O \left( \epsilon X, \epsilon Y, \epsilon U, \epsilon^2, X^2 U, U^3, X Y U  \right), \hfill
\end{aligned}
\end{equation}

\vspace{0.1in}
we find:

\begin{equation}
\label{eqA9}
\begin{aligned}
\tilde{a} & = A \left( \dfrac{\partial f_1}{\partial y} - \dfrac{F}{G} \dfrac{\partial f_1}{\partial u} \right)f_2 + 2 B f_2^2, \hfill \\
\tilde{b} & = \dfrac{A}{G} \dfrac{\partial f_1}{\partial u}, \hfill  \\
\tilde{c} & = \dfrac{\partial g_1}{\partial z}. \hfill
\end{aligned}
\end{equation}

where

\begin{equation}
\label{eqA10}
\begin{aligned}
A & = \frac{1}{2} \dfrac{\partial g_2}{\partial x} \dfrac{\partial^2 g_2}{\partial u^2}, \hfill \vspace{6pt} \\
B & = \dfrac{1}{4} \left[ \dfrac{\partial^2 g_2}{\partial u^2}\dfrac{\partial^2 g_2}{\partial y^2}  - \left( \dfrac{\partial^2 g_2}{\partial y \partial u} \right)^2 \right], \hfill \vspace{6pt} \\
G & = - \frac{1}{2} \dfrac{\partial^2 g_2}{\partial u^2}. \hfill
\end{aligned}
\end{equation}

Finally, we deduce:

\begin{equation}
\label{eqA11}
\begin{aligned}
\tilde{a} & = \frac{1}{2} \left[ f_2^2 \left( \dfrac{\partial^2 g_2}{\partial x_2^2 } \dfrac{\partial^2 g_2}{\partial y_2^2} - \left(\dfrac{\partial^2 g_2}{\partial x_2 \partial y_2} \right)^2 \right) + f_2 \dfrac{\partial g_2}{\partial x_1} \left( \dfrac{\partial^2 g_2}{\partial y_2^2 }\dfrac{\partial f_1}{\partial x_2} - \dfrac{\partial^2 g_2}{\partial x_2 \partial y_2 } \dfrac{\partial f_1}{\partial y_2} \right) \right] \hfill \\
\tilde{b} & = - \dfrac{\partial g_2}{\partial x_1}\dfrac{\partial f_1}{\partial y_2}, \hfill  \\
\tilde{c} & = \dfrac{\partial g_1}{\partial y_1}. \hfill
\end{aligned}
\end{equation}

\smallskip

This is the result we established in Sec. 2.7.

\end{article}
\end{document}